\documentclass[a4paper,UKenglish]{lipics-v2018}

\usepackage{microtype}
\usepackage{mathtools,amsmath} \usepackage{amssymb,amsthm}
\usepackage{graphicx} \usepackage{xspace} \usepackage{xypic}
\usepackage[numbers]{natbib} \usepackage{hyperref,color,lineno}
\usepackage[linesnumbered,ruled,vlined]{algorithm2e}

\graphicspath{{figures/}}
\title{Linear-Time Approximation Scheme for $k$-Means Clustering of Affine Subspaces%
}
\titlerunning{Linear-Time Approximation Scheme for $k$-Means Clustering of Affine Subspaces
} 

\author{
Kyungjin Cho}{POSTECH}{kyungjincho@postech.ac.kr}{}{}
\author{Eunjin Oh}{POSTECH}{eunjin.oh@postech.ac.kr}{}{}

\authorrunning{
K. Cho and E. Oh} 

\Copyright{
Kyungjin Cho and Eunjin Oh}


\supplement{}

\funding{
	This work was supported by the National Research Foundation of Korea (NRF) grant
	funded by the Korea government (MSIT) (No.2020R1C1C1012742).}

\acknowledgements{}

\EventEditors{John Q. Open and Joan R. Acces} \EventNoEds{2}
\EventLongTitle{42nd Conference on Very Important Topics (CVIT 2016)}
\EventShortTitle{CVIT 2016} \EventAcronym{CVIT} \EventYear{2016}
\EventDate{December 24--27, 2016} \EventLocation{Little Whinging,
	United Kingdom} \EventLogo{} \SeriesVolume{42} \ArticleNo{23}
\nolinenumbers 
\hideLIPIcs  

\theoremstyle{plain}

\graphicspath{{figures/}}

\newcommand{\cost}{\ensuremath{\textsc{cost}}}
\newcommand{\fd}{\ensuremath{\textsf{FD}}}
\newcommand{\pd}{\ensuremath{\textsf{PD}}}

\newcommand{\dom}{\ensuremath{\textsc{dom}}}
\newcommand{\opt}{\ensuremath{\textsc{opt}}}

\newcommand{\means}{\ensuremath{\textsf{Means}}}

\newcommand{\mmeans}{\ensuremath{\textsf{Means }}}

\newtheorem{claim}[theorem]{Claim}
\newcommand{\RN}[1]{%
	\textup{\uppercase\expandafter{\romannumeral#1}}%
}
\newcommand{\rn}[1]{%
	\textup{\expandafter{\romannumeral#1}}%
}

\SetKwInOut{Input}{input}\SetKwInOut{Output}{output}

\begin{document}
	
	\maketitle
	\begin{abstract}
		In this paper, we present a linear-time approximation scheme
		for $k$-means clustering of \emph{incomplete} data points
		in $d$-dimensional Euclidean space.
		An \emph{incomplete} data point with $\Delta>0$ unspecified entries 
		is represented as an  
		axis-parallel affine subspaces of dimension $\Delta$. 
		The distance between two incomplete data points is defined
		as the Euclidean distance between two closest points in the axis-parallel affine subspaces
		corresponding to the data points. 
		We present an algorithm for $k$-means clustering of axis-parallel affine
		subspaces of dimension $\Delta$ that
		yields an $(1+\epsilon)$-approximate solution in $O(nd)$ time.
		The constants hidden behind $O(\cdot)$ depend only on $\Delta, \epsilon$ and $k$. 
		This improves the $O(n^2 d)$-time algorithm by Eiben et al.~[SODA'21] by a factor of $n$.
	\end{abstract}

\section{Introduction}
Clustering is a fundamental research topic in computer science, which arises in various
applications~\cite{Jain-1999}, including pattern recognition and classification, data mining,
image analysis, and machine learning. 
In clustering, the objective is to group a set of data points into clusters so that
the points from the same cluster are similar to each other. 
Usually, input points lie in a high-dimensional space, and the similarity between two points is
defined as their distance. Two of the popular clusterings 
are $k$-median and $k$-means clusterings.
In the $k$-means clustering problem, 
we wish to partition a given point set 
into $k$ clusters to minimize the sum of squared distances of 
each point to its cluster center. 
Similarly, in the $k$-median clustering problem, 
we wish to partition a given point set 
into $k$ clusters to minimize the sum of distances of 
each point to its cluster center.

In this paper, we consider clustering for \emph{incomplete data points}. 
The analysis of incomplete data is a long-standing challenge in
practical statistics. 
There are lots of scenarios where entries of points of a given data set are  
incomplete~\cite{allison2001missing}. For instance, a few questions are left blank
on a questionnaire; weather records for a region omit the figures for one weather station for a short period because of a
malfunction; stock exchange data is absent for one stock on
one day because of a trading suspension. 
Various heuristic, greedy, convex optimization, statistical, or even ad hoc methods
were proposed throughout the years in different practical domains to handle missing data~\cite{allison2001missing}.

Gao et al.~\cite{gao2008analysis} introduced a geometric approach 
to deal with incomplete data points for clustering problems. 
An \emph{incomplete} point has one or more unspecified entries, which can be
represented as an axis-parallel affine subspace. 
The distance between two incomplete data points is defined
as the Euclidean distance between two closest points in the axis-parallel affine subspaces
corresponding to the data points. 
Since the distance between an axis-parallel affine subspace 
and a point 
is well-defined, the classical clustering problems such
as $k$-means, $k$-median, and $k$-center can be defined on
a set of axis-parallel affine subspaces. 

The $k$-center problem in this setting was studied by~\cite{gao2008analysis,10.1145/1868237.1868246,lee2013clustering}. 
Gao et al.~\cite{gao2008analysis,10.1145/1868237.1868246} focused on the $k$-center clustering for $k\leq 3$, 
and presented an approximation algorithm for 
the $k$-center clustering of axis-parallel affine subspaces. 
Later, Lee and Schulman~\cite{lee2013clustering} 
improved the running time of the algorithm by Gao et al., and then
presented an $O(nd)$-time approximation algorithm for the $k$-center clustering problem for a larger $k$. 
The constant hidden behind $O(\cdot)$ depends on $\Delta, \epsilon$ and $k$. 
Moreover, they showed that the running time of an 
approximation algorithm with any approximation ratio cannot be polynomial in even one of
$k$ and $\Delta$ unless P = NP, and thus the running time of their algorithm 
is almost tight. 

Very recently, Eiben et al.~\cite{EPTAS} presented an approximation algorithm
for the $k$-means clustering of $n$ axis-parallel affine subspaces of dimension $\Delta$.	
Their algorithm yields an $(1+\epsilon)$-approximate solution in $O(n^2d)$ time with probability $O(1)$. 
The constant hidden behind $O(\cdot)$ depends on $\Delta, \epsilon$ and $k$. 
Since the best-known algorithm for the $k$-center clustering in this setting  
runs in time linear in both $n$ and $d$ (but exponential in both $k$ and $\Delta$),  
it is a natural question if a similar time bound can be achieved for the $k$-means clustering.
In this paper, we resolve this natural question by presenting an $(1+\epsilon)$-approximation algorithm
for the $k$-means clustering problem running in time linear in $n$ and $d$.  

\subparagraph{Related work.}	
The $k$-median and $k$-means clustering problems for \emph{points} in $d$-dimensional Euclidean space
have been studied extensively.
Since these problems are NP-hard even for $k=2$ or $d=2$~\cite{aloise2009np,mahajan2012planar,MEGIDDO1990327},
the study of $k$-means and $k$-median clusterings have been devoted to obtain $(1+\epsilon)$-approximation
algorithms for these problems~\cite{10.1145/1824777.1824779,Chen-2009,Feldman-2011,Har-Peled2007-smaller,kumar2010linear}.
These algorithms run in time polynomial time in the input size if one of $k$ and $d$ is constant. 
Indeed, it is NP-hard to approximate Euclidean $k$-means clustering within a factor better than a certain constant larger than one~\cite{awasthi2015hardness}. That is, the $k$-means clustering problem
does not admit a PTAS for arbitrary $k$ and $d$ unless P=NP. 

Also, the clustering problems for \emph{lines} (which are not necessarily axis-parallel) also have been  studied~\cite{NEURIPS2019_6084e82a,5459200}. 
More specifically, Ommer and Malik~\cite{5459200} presented
a heuristic for $k$-median clustering of lines in three-dimensional space.
Later, Marom and Feldman~\cite{NEURIPS2019_6084e82a} presented an algorithm for computing a coreset of
size $dk^{O(k)}\log n/\epsilon^2$, which gives a polynomial-time $(1+\epsilon)$-approximation algorithm
for the $k$-means clustering of lines in $d$-dimensional Euclidean space. 

\subparagraph{Our results.}
We present an algorithm for $k$-means clustering of axis-parallel affine
subspaces of dimension $\Delta$ that
yields an $(1+\epsilon)$-approximate solution in
$2^{O(\frac{\Delta^4k}{\epsilon} (\log \frac{\Delta}{\epsilon} + k))}dn$ time 
with a constant probability. 
This improves the previously best-known algorithm by Eiben et al~\cite{EPTAS}, which takes 
$2^{O(\frac{\Delta^7k^3}{\epsilon}(\log \frac{k\Delta}{\epsilon} ))}dn^2$ time.
Since it is a generalization of the $k$-means clustering problem for points ($\Delta=0$),
it does not admit a PTAS for arbitrary  $k$ and $d$ unless P=NP. 
Furthermore, similarly to Lee and Schulman~\cite{lee2013clustering}, we show in Appendix~\ref{apd:hardness} that 
an approximation algorithm with any approximation ratio cannot run in polynomial time in even one of 
$k$ and $\Delta$ unless P = NP. 
Thus, the running time of our algorithm is almost tight.

	\section{Preliminaries}
	
	We consider points in $\mathbb{R}^d$ with missing entries in some coordinates. Let us denote the missing entry value by $\otimes$, and let $\mathbb{H}^d$ denote the set of elements of $\mathbb{R}^d$ where we allow some coordinates to take the value $\otimes$.  
	Furthermore, we call a point in $\mathbb{H}^d$ a \emph{$\Delta$-missing point} if 
	at most $\Delta$ of its coordinates have value $\otimes$.
	We use $[k]$ to denote the set $\{1,\dots, k\}$ for any integer $k\geq 1$.
	For any point $u\in \mathbb{H}^d$ and an index $i\in[d]$, we use $(u)_i$ to denote the entry of the $i$-th coordinate of $u$. 
	If it is understood in context, we simply use $u_i$ to denote $(u)_i$.
	Throughout this paper, we use $i$ {or $j$} to denote 
	an index of the coordinates of a point, and $t$ to denote an index of a sequence (of points or sets).
	We use $(u_t)_{t\in[k]}$ to denote a $k$-tuple consisting of $u_1,u_2,\ldots, u_k$. 
	
	\subparagraph{Distance between two $\Delta$-missing points.}
	The  \emph{domain} of a point $u$ in $\mathbb{H}^d$, denoted by $\dom(u)$, is defined as the set of coordinate-indices $i\in[d]$ with $(u)_i \neq\otimes$. 
	For a set $I$ of coordinate-indices in $[d]$, we say that $u$ is \emph{fully defined} on $I$ 
	if $\dom(u)\subseteq I$. Similarly, we say that $u$ is \emph{partially defined} on $I$ if $\dom(u)\cap I\neq\emptyset$. 
	For a set $P$ of points of $\mathbb{H}^d$ and a set $I$ of coordinate-indices in $[d]$, 
	we use $\fd(P,I)$ to denote the set of points of $P$ fully defined on $I$. Similarly, 
	we use $\pd(S,I)$ to denote the set of points of $P$ partially defined on $I$.
	The \emph{null} point is a point $p\in\mathbb{H}^d$ 
	such that $(p)_i = \otimes$ for all indices $i\in[d]$. 
	With a slight abuse of notation, we denote the null point by $\otimes$ if it is understood in context.
	Also, we sometimes use $I_t$ to denote $\dom(u_t)$ if it is clear in context. 
	
	Notice that a $\Delta$-missing point in $\mathbb{H}^d$ can be considered as a $\Delta$-dimensional
	affine subspace in $\mathbb{R}^d$. 
	The distance between two $\Delta$-missing points in $\mathbb{H}^d$ is defined as the 
	Euclidean distance between their corresponding $\Delta$-dimensional affine subspaces in $\mathbb{R}^d$.
	More generally, we define the \emph{distance} between two points $x$ and $y$ in $\mathbb{H}^d$ \emph{on} a set $I\subseteq [d]$ as 
	\[d_I(x,y)=\sqrt{\sum_{i\in I}|x_i-y_i|^2},\]
	where $|a-b|=0$ for $a=\otimes$ or $b=\otimes$ by convention. 
	
	\subparagraph{The $k$-Means clustering of $\Delta$-missing points.}
	In this paper, we consider the $k$-\emph{means clustering} of $\Delta$-missing points of $\mathbb{H}^d$. 
	As in the standard setting (for $\Delta=0$), we wish to partition a given point set $P$
	into $k$ clusters to minimize the sum of squared distances of 
	each point to its cluster center. 
	For any partition $(P_t)_{t\in[k]}$ of $P$ into $k$ clusters such that each cluster $P_t$ is associated
	with a cluster center $c_t\in\mathbb{R}^d$, 
	the \emph{cost} of the partition is 
	defined as the sum of squared distances of each point in $P$ to its cluster center.
	
	\medskip 
	To be more precise, we define the \emph{clustering cost} as follows. 
	For a set $P\subset \mathbb{H}^d$ and a $\Delta$-missing point $y$, we use $\cost(P,y)$ to denote the sum of squared distances of each point in $P$ to $y$. We also define the cost on a coordinate set $I\subseteq[d]$,
	denoted by $\cost_I(P,y)$, 
	as the sum of squared distances on $I$ between the points in $P$ and their cluster centers.
	That is, $\sum_{x\in P} d_I(x,y)^2$. For convention, $\cost_i(P,y)=\cost_{\{i\}}(P,y)$ for $i\in [d]$.
	The clustering cost $\cost((P_t)_{t\in[k]}, (c_t)_{t\in[k]})$ of clustering $((P_t)_{t\in[k]}, (c_t)_{t\in[k]})$ is
	defined as $\sum_{t\in[k]}\cost (P_t, c_t)$. 
	
	\medskip
	Now we consider two properties of an optimal clustering $((P^*_t)_{t\in[k]}, (c^*_t)_{t\in[k]})$ that minimizes the clustering cost,
	which will  be frequently used throughout this paper. 
	For each cluster $P_t^*$, $\cost(P^*_t, c_t)$ is minimized when $c_t$ is the \emph{centroid} of $P_t^*$~\cite{EPTAS}. 
	That is, $c_t^*$ is the centroid of $P_t^*$. 
	For a set $P$ of points in $\mathbb{H}^d$, 
	the \emph{centroid of $P$}, denoted by $c(P)$, is defined as 
	\[ (c(P))_i=\left\{\begin{array}{ll} \otimes & \textnormal{if}\  \pd(P,i)=\emptyset,
	\\ \frac{\sum_{u\in \pd(P,i)}u_i}{|\pd(P,i)|} &\textnormal{otherwise.}
	\end{array}\right.\]

	Also, the clustering cost is minimized when $(P_t^*)_{t\in[k]}$ forms
	the Voronoi partition of $P$ induced by $(c_t^*)_{t\in[k]}$~\cite{EPTAS}. 
	That is, $(P_t^*)_{t\in[k]}$ is the partition of $P$ into $k$ clusters in such a way that 
	$c_t^*$ is the closest cluster point from any point $p$ in $P_t^*$. 
	\medskip

	\subparagraph{Sampling.}
	Our algorithm uses random sampling to compute an approximate $k$-means clustering. 
	Lemma~\ref{lem:superset} is a restatement of~\cite[Lemma~2.1]{10.1145/1824777.1824779}, and
	Lemma~\ref{lem:initial} is used in~\cite{EPTAS} implicitly.
	Since Lemma~\ref{lem:initial} is not explicitly stated in~\cite{EPTAS}, we give
	a sketch of the proof in Appendix~\ref{apd:proof}.
	%
	%

	
	\begin{lemma}[{\cite[Lemma~2.1]{10.1145/1824777.1824779}}]\label{lem:superset}
		Assume that we are given a set $P$ of points in $\mathbb{H}^d$, an index $i\in[d]$, and an approximation factor $\alpha>0$. 
		Let $Q$ be a subset of $P$ with $|\pd(Q,i)|\geq c|P|$ for some constant $c$, which is not given
		explicitly. 
		Then we can compute a point $x$ of $\mathbb{R}$ in $O(|P|d m_{\alpha,\delta})$ time satisfying  with probability $\frac{1-\delta}52^{\Omega(-m_{\alpha,\delta}\log (\frac 1c m_{\alpha,\delta}))}$ that 
		\[\cost_i(Q, x) \leq (1+\alpha) \cost_i(Q, c(Q)),\]
		where $m_{\alpha,\delta} \in O(1/(\alpha\delta))$. 
		\end{lemma}
	
	\begin{lemma}[\cite{EPTAS}]\label{lem:initial} 
		Assume that we are given a set $P$ of $\Delta$-missing points in $\mathbb{H}^d$ and an approximation factor $\alpha>0$.
		Let $Q$ be a subset of $P$ with $|Q|\geq c|P|$ for some constant $c$ with $0<c<1$, which is not given explicitly. Then 
		we can compute a $\Delta$-missing point $u\in\mathbb{H}^d$ in $O(|P|d \lambda)$ time satisfying with probability  $\frac{c^{8(\Delta+1)\lambda+1}}{4(4\Delta)^{8\Delta\lambda}}$ that 
		\[\cost_{I}(Q,u)<(1+\alpha) \cost_I(Q, c(Q)),\]
		where $I$ denotes the domain of $u$, and $\lambda=\max\{(\frac{3}{\alpha})^{1/(2\Delta)},(128\Delta^3)^{1/(2\Delta)}\}$.
	\end{lemma}
\begin{proof}[Proof (Sketch)]
	We randomly select a point $p$ from $P$, and then define $u$ as follows so that
	$\dom(u)=\dom(p)$. 
	To do this, we choose a random sample $T$ of size $8\lambda$ from $P$. 
	For each coordinate-index $i\in \dom(p)\cap \dom(c(T))$, we set $(u)_i = (c(T))_i$. 
	For each coordinate-index $j\in \dom(p)-\dom(c(T))$, 
	we choose a random sample $T_j$ of size $8\lambda$ from $P$,
	and set $(u)_j= (c(T_j))_j$.
	
	Eiben et al. showed that $\dom(u)=\dom(p)$  and $\cost_{I}(Q,u)<(1+\alpha)\cost_{I}(Q,c(Q))$ 
	with probability at least $\frac{c^{8(\Delta+1)\lambda+1}}{4(4\Delta)^{8\Delta\lambda}}$,
	where $I$ denotes the domain of $u$. 
	Details of the analysis can be found in the proof of~\cite[Lemma~17]{EPTAS}.
\end{proof}

\section{Overview of the Algorithm}
To describe our contribution, we first briefly describe a $(1+\epsilon)$-approximation algorithm for $k$-means clustering for points in $d$-dimensional Euclidean space given by Kumar et al.~\cite{kumar2010linear}. 
Let $P$ be a set of $n$ points in $d$-dimensional Euclidean space, and $((P_t^*)_{t\in[k]}, (c_t^*)_{t\in[k]})$ be an optimal $k$-means clustering for $P$. 

\subparagraph{Sketches of~\cite{10.1145/1824777.1824779} and~\cite{kumar2010linear}.}
The algorithm of Kumar et al.~\cite{kumar2010linear} consists of several phases of two types: sampling phases and pruning phases. 
Their idealized strategy is as follows. 
At the beginning of a phase, it decides the type of the phase by computing the index $t$ that maximizes $|P_t^*|$. 
If the cluster center of $P_t^*$ has not been obtained, 
the algorithm enters the sampling phase. 
this algorithm picks a random sample of a constant size from $P$, 
and hopefully this sample would contain enough random samples from $P_t^*$. 
Then one can compute a good approximation $c_t$ to $c_t^*$ using Lemma~\ref{lem:superset}. 

If it is not the case, the algorithm enters a pruning phase, and 
it assigns each point in $R$ to its closest cluster if their distance is at most $L$, 
where $L$ denotes the smallest distance between two cluster centers we have been obtained so far.   
They repeat this until all cluster centers are obtained, and finally 
obtain a good approximation to $(P_t^*)_{t\in[k]}$. 

However, obviously, it is hard to implement this idealized strategy. 
To handle this, they try all possibilities (for both pruning and sampling phases and for all indices $t\in[k]$ to be updated for sampling phases), and return the best solution found this way. 
Kumar et al.~\cite{kumar2010linear} showed that their algorithm runs in $O(2^{(k/\epsilon)^{O(1)}} dn)$  time,
and returns an $(1+\epsilon)$-approximate $k$-means clustering with probability $1/2$. 
Later, Ackermann et al.~\cite{10.1145/1824777.1824779} gave a tighter bound on the running time of this algorithm. 

\subparagraph{Sketch of Eiben et al.~\cite{EPTAS}}
To handle $\Delta$-missing points, Eiben et al. generalized the algorithm in~\cite{kumar2010linear}.
Their idealized strategy (using the counting oracle) can be summarized as follows. 
It maintains $k$ centers $(u_t)_{t\in[k]}$, which are initially set to the null points.
In each sampling phase, it obtains one (or at least $[d]-\Delta$) coordinate of one of the centers.

At the beginning of a phase, it decides the type of the phase by computing the index $t$ that maximizes  $|\pd(P_t^*, [d]-I_t)|$. 
A sampling phase happens if $|\pd(P_t^*, [d]-I_t)| > c|R|$, where $R$ denotes the number of points
which are not yet assigned to any cluster. In this case, a random sample of constant size from $R$ 
would contain enough random samples from $|\pd(P_t^*, j)|$ with $j\in [d]-I_t$.
Thus, using the random sample, one can obtain a good approximation to $(c_t^*)_j$.  

Otherwise, a pruning phase happens. In a pruning phase,
the algorithm assigns points which are not yet assigned to any cluster
to clusters. Here, a main difficulty is that 
even though the distance between a point $p$ in $R$ and its closest center $u_t$ is at most $L$,
where $L$ denotes the distance between two cluster centers,
it is not necessarily that $p\in P_t^*$. 
They resolved this in a clever way by ignoring $\Delta$ coordinates for comparing
the distances from two cluster centers. 

\subparagraph{Comparison of our contribution and Eiben et al.~\cite{EPTAS}.}
Our contribution is two-fold: the dependency on $n$ decreases to $O(n)$ from $O(n^2)$, and
the dependency of $\Delta$ and $k$ decreases significantly. 

First, the improvement on the dependency of $n$ comes from 
introducing a faster and simpler procedure  for a pruning phase.
In the previous algorithm, it cannot be guaranteed that
a constant fraction of points of $R$ is removed from $R$. 
This yields the quadratic dependency of $n$ in their running time. 
We overcome the difficulty they (and we) face in a pruning phase
in a different way. 
For each subset $T$ of $[k]$, we consider the set $S_T$ of points $x\in R$
such that $\dom(x)\subset I_t$ for every $t\in T$ and $\dom(x)\not\subset I_{t'}$
for every $t'\notin T$. 
Then $S_T$'s for all subsets $T\in[k]$ form a partition of $R$. 
In a pruning phase, we choose the set $S_T$ that maximizes $|S_T|$. 
We show that the size of $S_T$ is at least a constant fraction of $|R|$ (unless we enter the sampling phase).
Moreover, in this case, 
if the distance between a point $p$ in $S_T$ and its closest center $u_t$ is at most $L$,
where $L$ denotes the distance between two cluster centers,
it holds that $p\in P_t^*$. 

Second, the improvement on the dependency of $\Delta$ and $k$ comes from
using the framework of Ackermann et al.~\cite{10.1145/1824777.1824779} to analyze the approximation factor
of the algorithm while Eiben et al.~\cite{EPTAS} uses 
the framework of Kumar et al.~\cite{kumar2010linear}.

\section{For $2$-Means Clustering}\label{sec:for_2_means_clustering}
In this section, we focus on the case that $k=2$, and in the following section,
we show how to generalize this idea to deal with a general constant $k>2$.

\subsection{Algorithm Using the Counting Oracle}\label{sec:two-algo}
In this section, we first sketch an algorithm for 2-means clustering assuming that we can access the \emph{counting oracle}. 
Let $(P_1^*,P_2^*)$ be an optimal 2-clustering for $P$, and $c_1^*$ and $c_2^*$ be the centroids of $P_1^*$ and $P_2^*$, respectively.
The \emph{counting oracle} takes a subset $Q$ of $P$ and a cluster-index $t=1,2$ as input, and then returns 
the number of points in $Q\cap P_t^*$. 
In Section~\ref{sec:refined}, we will modify this algorithm so that the use of the counting oracle can be avoided
without increasing the approximation factor. 
The assumption made in this subsection makes it easier to analyze the approximation factor of the algorithm. 

The algorithm consists of several phases of two types: a \emph{sampling phase} or a \emph{pruning phase}. 
Initially, we set $u_1, u_2$ so that $(u_t)_i=\otimes$ for all $t=1,2$ and $i\in[d]$.
In a sampling phase, we obtain values of $(u_t)_j$ for indices $t=1,2$ and $j\in[d]$ which were set to $\otimes$.  
Also, we assign points of $P$ to one of the two clusters in sampling and pruning phases. 
The pseudocode of the algorithm is described in Algorithm~\ref{algo:2_means:ideal}. 

At the beginning of each phase, we decide the type of the current phase. 
Let $t$ be the cluster-index that maximizes $|\pd(R\cap P_t^*,[d]-I_t)|$, where $R$ is the set of points of $P$ which are
not assigned to any cluster. 
The one of the following cases always happen: $|\pd(R\cap P_t^*,[d]-I_t)|\geq c|R|$ or $\sum_{t'=1,2}|\fd(R\cap P_{t'}^*,I_{t'})|\geq c|R|$
for a constant $c<1/4$, which will be specified later.\footnote{We will set $\alpha=\epsilon/3$ and $c=\frac\alpha{64(\Delta+1)^2}$}

\begin{algorithm}[t]
	\SetAlgoLined
	\Input{A set $P$ of $\Delta$-missing points in the plane}
	\Output{A $(1+\epsilon)$-approximate 2-means clustering for $P$}
	
	\SetNoFillComment 
	
	$R \gets P$ and $P_1, P_2 \gets\emptyset$\;
	Initialize $u_1,u_2$ so that $(u_t)_i=\otimes$ for all $i\in[d]$ and all $t\in[2]$\;
	\While{$R\neq\emptyset$} {	
		Let $t$ be the cluster-index that maximizes $|\pd(R\cap P_t^*,[d]-I_t)|$\;
		
		\If{$|\pd(R\cap P_t^*,[d]-I_t)|\geq c|R|$}{
			\tcc*[l]{sampling phase}
			{\If{$I_t=\emptyset$}
				{
					$u_t \gets$ the $\Delta$-missing point obtained from Lemma~\ref{lem:initial}\;
				}
				\Else
				{	Let $j$ be the coordinate-index in $[d]-I_t$ that maximizes $\pd(R\cap P_t^*, j)$\;
					$(u_t)_j \gets$ The value obtained from Lemma~\ref{lem:initial}\;
				}
				Add the points in $\fd(R,I_1\cap I_2)$ closer to $u_1$ than to $u_2$ to $P_1$\;
				Add the points in $\fd(R,I_1\cap I_2)$ closer to $u_2$ than to $u_1$ to $P_2$\;	
				$R \gets R-\fd(R,I_1\cap I_2)$\;	
		}}
		
		\Else{
			\tcc*[l]{pruning phase}
			$t \gets$ the cluster-index that maximizes $|\fd(R, I_t)|$\;
			$B\gets$ The first half of $\fd(R,I_t)$ sorted in ascending order of distance from $u_t$\;
			Assign the points in $B$ to $P_t$\;
			$R \gets R-B$\;
		}
		
	}
	\KwRet{$(u_1,u_2)$}
	\caption{\textbf{\textsf{Idealized 2-Means}}}\label{algo:2_means:ideal}
\end{algorithm}
\subparagraph{Sampling phase.} If the first case happens, we enter a sampling phase.
Let $\alpha$ be a constant, which will be specified later, which is an approximation factor used in Lemmas~\ref{lem:superset} and~\ref{lem:initial} for sampling.  
If $I_t$ is empty, we replace $u_t$ with a $\Delta$-missing point in $\mathbb{H}^d$  obtained from Lemma~\ref{lem:initial}.  
If $I_t$ is not empty, then it is guaranteed that $|I_t|$ is at least $d-\Delta$. 
We compute the coordinate-index $j$ in $[d]-I_t$ that maximizes $|\pd(R\cap P_t^*,j)|$
using the counting oracle. Clearly, $(u_t)_j=\otimes$ and 
$|\pd(R\cap P_t^*,j)|$ is at least $c|R|/\Delta$.
Then we replace $(u_t)_j$ with a value obtained from Lemma~\ref{lem:superset}.
At the end of the phase, we check if $\fd(R, I_1\cap I_2)$ is not empty.
If it is not empty, we assign those points to their closest cluster centers. 

\subparagraph{Pruning phase.} 
Now consider the second case. 
In this case, we enter a pruning phase. Instead of obtaining a new coordinate value of $u_t$, 
we assign points of $R$ to one of the clusters as follows. 
Let $t'$ be the cluster-index that maximizes $|\fd(R, I_{t'})|$. 
Among the points of $\fd(R,I_{t'})$, we choose $|\fd(R,I_{t'})|/2$ points closest to $u_{t'}$, 
and assign them to $u_{t'}$. 

\subparagraph{}
The algorithm terminates when every point of $P$ is assigned to one of the clusters.
Let $P_t$ be the set of points of $P$ assigned to $u_t$ for $t=1,2$.
Also, let $(c_1,c_2)$ be $(u_1,u_2)$ when the algorithm terminates. 
Notice that $(P_1, P_2)$ is not necessarily the Voronoi partition induced by $(c_1,c_2)$
because of the pruning phase. 
Also, by construction, a point $p$ of $P$ is assigned to $u_t$ for a cluster-index $t$ only when
it is fully defined on $I_t$ at the moment.

\subsection{Approximation Factor}\label{sec:two-analysis}
In this section, we analyze the approximation factor of the algorithm in Section~\ref{sec:two-algo}. 
To do this, 
let $\mathcal{S}$ be the sequence of sampling phases happened during the execution of the algorithm. 

\begin{lemma}\label{inv: invariants}
	At any time during the execution,  
	$|I_t|\geq d-\Delta$ or $|I_t|=0$ for $t=1,2$.
\end{lemma}
\begin{proof}
	Initially, all coordinates of $u_1$ and $u_2$ are set to $\otimes$, and thus $I_1=I_2=\emptyset$.
	The cluster centers $u_1$ and $u_2$ are updated during the sampling phases only.
	If $I_t$ is empty for $t=1,2$, we use Lemma~\ref{lem:initial}, and thus $u_t$ is updated to
	a $\Delta$-missing point. 
	If $I_t$ is not empty, we use Lemma~\ref{lem:superset}, and increase $|I_t|$ by one.
	Therefore, the lemma holds. 
	%
	%
\end{proof}

\begin{corollary}\label{cor:leftstep}
	The size of $\mathcal{S}$ is at most $2\Delta+2$. 
\end{corollary}

Now we show that the clustering cost induced by the points assigned to incorrect clusters
during the pruning phases is small compared to the clustering cost induced by the points assigned to correct
clusters. 

Consider two consecutive sampling phases $s$ and $s'$, and consider
the pruning phases lying between them. 
During this period, $u_1$ and $u_2$ remain the same. 
Let $R$ be the set of points of $P$ which were not yet assigned to any cluster at the \emph{beginning} of this period. 
In each iteration during this period, points of $R$ are assigned to one of the two clusters.
Let $R^{(x)}$ be the set of points of $P$ which are not yet assigned to any cluster
at the \emph{beginning} of the $x$th pruning phase (during this period).
Let $X_t^{(x)}=\fd(R^{(x)}, I_t)$ for $t=1,2$.  
Also, let $A_t^{(x)}$ denote the set of points of $P$ assigned to $u_t$ at the $x$th iteration.
By construction, either $A_1^{(x)}=\emptyset$ or $A_2^{(x)}=\emptyset$, but not both.

Let $\mathcal{X}_t$ be the increasing sequence of indices $x$ of $[N]$
with $A_t^{(x)}\neq\emptyset$ for $t=1,2$, where $N$ denotes the number of pruning phases 
lying between the two consecutive sampling phases we are considering.
To make the description easier, 
let $A_t^{(N+1)} = X_t^{(x_t)}$, where $x_t$ is the last index of $\mathcal{X}_t$.
Then we add $N+1$ at the end of $\mathcal{X}_1$
and $\mathcal{X}_2$. 

We first analyze the clustering cost induced by the points of $P_1^*$ (and $P_2^*$) assigned to $u_2$ (and $u_1$)
during the pruning phases. 
To do this, we need the following three technical claims, which will be used to prove Lemma~\ref{claim:wrong_assigned}.

\begin{claim}\label{claim:2-base} For any index $x\in[N]$, we have the following. 
	\[|X_1^{(x)}\cap P_2^*|+|X_2^{(x)}\cap P_1^*| < 2c(\Delta+1)(|X_1^{(x)}|+|X_2^{(x)}|).\]
\end{claim}
\begin{proof}
	Assume to the contrary that $|X_1^{(x)}\cap P_2^*|+|X_2^{(x)}\cap P_1^*|\geq 2c(\Delta+1)(|X_1^{(x)}|+|X_2^{(x)}|)$.
	We show that there is a cluster-index $t$ 
	such that $\pd(P_t^* \cap R^{(x)}, [d]-I_t)$ has $c|R^{(x)}|$ points, which contradicts that the $x$th phase is 
	a pruning phase. 	
	
	To see this, observe the following. Let $R' = R^{(0)}- (X_1^{(0)} \cup X_2^{(0)})$.
	That is, $R'$ is the set of points which are not yet assigned to any cluster at the beginning of
	these pruning phases, but whose domains are fully defined on neither $I_1$ nor $I_2$. 
	Therefore, during the pruning phases (in this period), no point of $R'$ is assigned to any cluster. 
	\begin{align*}
	\sum_{t=1,2}|\pd(P_t^*\cap R^{(x)}, [d]-I_t)|&=
	\sum_{t=1,2}|\pd(P_t^* \cap (R' \cup X_1^{(x)} \cup X_2^{(x)}), [d]-I_t)|\\
	&= |P_1^*\cap X_2^{(x)}| + |P_2^* \cap X_1^{(x)}| +\sum_{t=1,2}|R' \cap P_t^*| \\
	&\geq 2c(\Delta+1)(|X_1^{(x)}|+|X_2^{(x)}|+|R'|) \\
	&= 2c(\Delta+1)|R^{(x)}|
	\end{align*}
	The first and last equalities hold since $R^{(x)}$ is decomposed by $R'$, $X_1^{(x)}$, and $X_2^{(x)}$. The second equality holds because $X_t^{(x)}\subset \fd(R^{(x)},I_t)$ for $t=1,2$ and 
	no point in $R'$ is fully defined on $I_t$ for any cluster-index $t=1,2$. 
	The third inequality holds by the fact that $2c(\Delta+1)<1$ and 
	the assumption we made at the beginning of this proof.
	
	Let $t$ be the cluster-index that maximizes $|\pd(P_t^* \cap R^{(x)}, [d]-I_t)|$. 
	Then we have $|P_t^*\cap \pd(R^{(x)}, [d]-I_t)| \geq c(\Delta+1)|R^{(x)}|$. 
	Therefore, the lemma holds. 
\end{proof}

\begin{claim}\label{claim:2-fraction}
	For any consecutive indices $x'$ and $x$ in $\mathcal{X}_1$ with $x'<x$, 
	$\frac{|A_1^{(x')}\cap P_2^*|}{|A_1^{(x)} \cap P_1^*|} \leq \frac{16c(\Delta+1)}{1-8c(\Delta+1)}$.
\end{claim}
\begin{proof}
	We first show that $|A_1^{(x)}\cap P_2^*|$ is at most $8c(\Delta+1){|A_1^{(x)}|}$ as follows.
	\begin{align*}
	|A_1^{(x)}\cap P_2^*|&\leq |X_1^{(x')}\cap P_2^*|\\
	&<2c(\Delta+1)(|X_1^{(x')}|+|X_2^{(x')}|) \\
	&\leq 4c(\Delta+1)|X_1^{(x')}|\\
	&\leq 8c(\Delta+1)|A_1^{(x)}|.
	\end{align*}
	The first inequality holds because $A_1^{(x)}$ is a subset of $X_1^{(x')}$.
	The second inequality holds by Claim~\ref{claim:2-base}, and the third inequality holds
	since $|X_1^{(x')}|>|X_2^{(x')}|$. If it is not the case,
	the algorithm would assign a half of $X_2^{(x')}$ to $u_2$, and thus $A_1^{(x)}=\emptyset$, which 
	contradicts that $x\in \mathcal{X}_1$. The last equality holds since $|X_1^{(x')}|\leq 2|A_1^{(x)}|$.
	
	Now we give a lower bound of $|A_1^{(x)}\cap P_1^*|$ as follows. 
	The second inequality holds due to the upper bound of $|A_1^{(x)}\cap P_2^*|$ stated above. 
	\begin{align*}
	|A_1^{(x)}\cap P_1^*|&=|A_1^{(x)}|-|A_1^{(x)}\cap P_2^*|\\
	&\geq 	|A_1^{(x)}|- 8c(\Delta+1) |A_1^{(x)}|\\
	&=	(1- 8c(\Delta+1)){|A_1^{(x)}|}
	\end{align*}
	
	By combining the upper and lower bounds, we can obtain the following. 
	The last inequality holds because $|A_1^{(x')}| \leq 2|A_1^{(x)}|$. 
	\[\frac{|A_1^{(x')}\cap P_2^*|}{|A_1^{(x)} \cap P_1^*|} \leq \frac{8c(\Delta+1)|A_1^{(x')}|}{1-8c(\Delta+1)|A_1^{(x)}|} 
	\leq \frac{16c(\Delta+1)}{1-8c(\Delta+1)}.\tag*{\qedhere}\]
\end{proof}

\begin{claim}\label{claim:2-last}
	Let $A_t=\cup_{x\in [N]} A_t^{(x)}$ for $t=1,2$. 
	\begin{itemize}
		\item $\cost(A_1 \cap P_2^*,{c_1})\leq 
		32c(\Delta+1)\cost(\fd({R^{(0)}}\cap P_1^*,I_1),c_1)$, and
		\item $\cost(A_2\cap P_1^*,{c_2})\leq 
		32c(\Delta+1)\cost(\fd({R^{(0)}}\cap P_2^*,I_2),c_2)$. 
	\end{itemize}
\end{claim}
\begin{proof}
	We prove the first inequality only. The other is proved analogously. 
	Recall that $\mathcal{X}_1$ is the increasing sequence of indices $x$ of $[N+1]$
	with $A_1^{(x)}\neq\emptyset$, and the last index of $\mathcal{X}_1$ is $N+1$. 
	Consider two consecutive indices $x'$ and $x$ in $\mathcal{X}_1$ with $x'<x$. 
	By definition, 
	for any index $x\in[N]$, $c_1$ is closer to any point of $A_1^{(x')}$ than
	to any point of $A_1^{(x)}$. Therefore, we have		
	\[
	\frac{\cost(A_1^{(x')}\cap P_2^*, c_1)}{|A_1^{(x')}\cap P_2^*|}\leq\frac{\cost(A_1^{(x)}\cap P_1^*,c_1)}{|A_1^{(x)}\cap P_1^*|}.\] 
	By Claim~\ref{claim:2-fraction}, we have
	\[\cost(A_1^{(x')}\cap P_2^*, c_1)  \leq\frac{16c(\Delta+1)}{1-8c(\Delta+1)}\cost(A_1^{(x)}\cap P_1^*,c_1).\]
	
	Note that the points of $P$ assigned to $u_1$ during these consecutive pruning phases are 
	exactly the points in the union of $A_1^{(x)}$'s for all indices $x\in [N]$. 
	Moreover, $A_1^{(x)}$'s  are pairwise disjoint.
	Also, $\frac{16c(\Delta+1)}{1-8c(\Delta+1)}\leq 32c(\Delta+1)$. Therefore, 
	\[\cost(A_1 \cap P_2^*)\leq 
	32c(\Delta+1)\cost(\fd({R^{(0)}}\cap P_1^*,I_1),c_1).\]
	The other inequality can be proved symmetrically. 
\end{proof}

We are ready to analyze the clustering cost induced by the points incorrectly assigned during the pruning phases. 
To do this, we define several notations. 
Recall that $\mathcal S$ is the sequence of sampling phases happened in the course of the algorithm.
Here, we represent a sampling phase as the pair $(t,I)$,
where $t$ is the cluster-index considered in the sampling phase, and $I$ is the set of indices $i$
such that $(u_t)_i$ is obtained during the sampling phase. 
For two sampling phases $s$ and $s'$ in $\mathcal{S}$,
we use $s\preceq s'$ if $s$ comes before $s'$ in $\mathcal{S}$ or equals to $s'$. 

For a sampling phase $s\in\mathcal{S}$, let $R_t^s$ be the set of points of $P_t^*$ which are not assigned at the 
\emph{beginning}
of the sampling phase $s \in \mathcal{S}$. Also, let $t^s=t$ and $I^s=I$ for $s=(t,I)$. 
For a cluster-index $t$ and a sampling phase $s\in\mathcal{S}$, let $A_t^s$ denote the set of points assigned to $P_t$ during the pruning phases lying between $s$ and $s'$,
where $s'$ is the sampling phase coming after $s$ in $\mathcal{S}$. 
Furthermore, let $I_t^s$ denote $\dom(u_t)$ we have during these pruning phases.

\begin{lemma}\label{claim:wrong_assigned} The following holds,
	\[\sum_{s \in \mathcal{S}} \cost(A_1^s\cap P_2^*,c_1)+\sum_{s\in\mathcal{S}} \cost(A_2^s\cap P_1^*,c_2) \leq (64c(\Delta+1)^2)\sum_{\substack{s\in\mathcal{S}}} \cost_{I^s}(R_{t^s}^s,c_{t^s}).\]
\end{lemma}
\begin{proof}
	
	Let $\mathcal{S}_t$ be the set of sampling phases $s$ in $\mathcal{S}$ with $t^s=t$.
	\begin{align*}
	\sum_{s \in \mathcal{S}} \cost(A_1^s\cap P_2^*,c_1)
	&\leq 32c(\Delta+1) \sum_{\substack{s\in\mathcal{S}}} \cost(\fd(R_1^s \cap P_1^*,I_1^s),c_1) \\
	&\leq 32c(\Delta+1) \sum_{\substack{s\in\mathcal{S}}} \sum_{\substack{ s'\preceq s \\ t^{s'}=1}} 
	\cost_{I^{s'}}(\pd(R_1^s\cap P_1^*,I^{s'}),c_1) \\
	&= 32c(\Delta+1) \sum_{\substack{ s'\in \mathcal{S}_1} } \sum_{s'\preceq s}\cost_{I^{s'}}(\pd(R_1^s\cap P_1^*,I^{s'}),c_1) \\
	&\leq 32c(\Delta+1) \sum_{\substack{ s'\in \mathcal{S}_1}} \sum_{s' \preceq s} \cost_{I^{s'}}(\pd(R_1^{s'},I^{s'}),c_1) \\
	&\leq 64c(\Delta+1)^2 \sum_{\substack{ s\in \mathcal{S}_1}}  \cost_{I^s}(\pd(R_1^s,I^s),c_1) \\
	&= 64c(\Delta+1)^2 \sum_{\substack{ s\in \mathcal{S}_1}} \cost_{I^s}(R_1^s ,c_1),
	\end{align*}
	The first inequality holds by Claim~\ref{claim:2-last}. The second one holds since $s'\preceq s$ for two sampling phases $s'$ and $s$ in $\mathcal{S}$ with $I_1^{s'}\subset I_1^s$. 
	The third equality holds since it changes only the ordering of summation. 
	The fourth inequality holds since  for two sampling phases $s$ and $s'$ $\mathcal{S}$, $R_1^s\subset R_1^{s'}$ if $s'\preceq s$. 
	The fifth inequality holds since the number of sampling phases is at most $2(\Delta+1)$
	The last equality holds by definition of $\cost(\cdot)$.
	
	Analogously, we have
	\[
	\sum_{s \in \mathcal{S}} \cost(A_2^s\cap P_1^*,c_2) \leq 64c(\Delta+1)^2 \sum_{s \in \mathcal{S}_2} \cost_{I^s}(R_2^s,c_2).
	\]
	Finally, we have
	\begin{equation*}
	\sum_{s \in \mathcal{S}} \cost(A_1^s\cap P_2^*,c_1)+\sum_{s\in\mathcal{S}} \cost(A_2^s\cap P_1^*,c_2) \leq 64c(\Delta+1)^2 \sum_{\substack{s\in\mathcal{S}}} \cost_{I^s}(R_{t^s}^s,c_{t^s}). \tag*{\qedhere}
	\end{equation*}
\end{proof}

We now analyze the clustering cost induced by $P_t$ excluding the points incorrectly assigned to a cluster during the
pruning phases. 
That is, we give an upper bound on the clustering cost induced by the points assigned during the sampling phases
and the points assigned correctly during the pruning phases. 
Let $T_t$ denote the set of points of $P_t^*$ assigned to $u_t$ during the pruning phases 
and $S_t$ denote the set of points assigned to $u_t$ during the sampling phases.

\begin{lemma}\label{claim:correctly_assigned}
	$\displaystyle\cost(S_1\cup T_1, c_1) + \cost(S_2 \cup T_2, c_2)\leq \sum_{\substack{s\in\mathcal{S}}} \cost_{I^s}(R_{t^s}^s,c_{t^s}).$
\end{lemma}
\begin{proof}
	Let $S= S_1\cup S_2$. By construction, $(S_1, S_2)$ is the Voronoi partition induced by $(c_1,c_2)$. 
	To see this, observe that 
	a value of a coordinate of $u_t$ is not changed further once it is obtained.
	Furthermore, in a sampling phase, we assign a point $p$ to a cluster only when $\dom(p) \in I_1\cap I_2$. 
	Therefore, 	we have
	\begin{equation}\label{eqn:voronoi}
	\cost(S_1, c_1) + \cost(S_2, c_2)
	\leq \cost(S\cap P_1^*, c_1) + \cost(S\cap P_2^*, c_2).
	\end{equation}
	
	Also, notice that $\pd((S\cup T_t)\cap P_t^*, I)$ is a subset of $R^s_{t}$ for a sampling phase $(t,I)=s\in\mathcal{S}$. 
	To see this, consider a point $p$ in $\pd((S\cup T_t)\cap P_t^*, I)$.
	Let $I_t$ is the domain of $u_t$ at the moment when $p$ is assigned to a cluster. 
	If $p$ is assigned during a sampling phase,
	by construction, we have $\dom(p) \subseteq I_1\cap I_2$. If $p$ is assigned to $u_t$ during a pruning phase,
	we have $\dom(p)\subseteq I_t$.  
	By combining all properties mentioned above, we have 
	\begin{align*}
	\cost(S_1\cup T_1, c_1) + \cost(S_2\cup T_2, c_2) 
	&\leq \sum_{t=1,2}\cost((S\cup T_t)\cap P_t^*, c_t)\\
	&= \sum_{\substack{s\in\mathcal{S}}}\cost_{I^s}(\pd((S\cup T_{t^s})\cap P_{t^s}^*, I^s),c_{t^s})\\
	&\leq \sum_{\substack{s\in\mathcal{S} }}\cost_{I^s}(R_{t^s}^s,c_{t^s}).
	\end{align*}
	
	The first inequality comes from~(\ref{eqn:voronoi}) and the fact that $T_t \subseteq P_t^*$, and 
	the second equality holds by the definition of $\cost(\cdot)$.
	The last inequality holds since $\pd((S\cup T_t)\cap P_t^*, I)$ is a subset of $R_t^s$ for $s=(t,I)$. 
\end{proof}

Therefore, it suffices to analyze an upper bound of the sum of $\cost_{I^s}(R_{t^s}^s,c_{t^s})$ for all sampling phases $s$ in $\mathcal S$.  
We show that it is bounded by $(1+\alpha)\opt_2(P)$ with a constant probability. 

\begin{lemma}\label{claim:alpha_opt} 
	$\displaystyle \sum_{\substack{s\in\mathcal{S}}} \cost_{I^s}(R_{t^s}^s,c_{t^s})\leq (1+\alpha)\opt_2(P),$\\
	with a probability at least $p^2q^{2\Delta}$, where $q$ and $p$ are the probabilities in Lemmas~\ref{lem:superset} and~\ref{lem:initial}.
\end{lemma}
\begin{proof}
	For convenience, let $R^s=R_{t^s}^s$.
	The algorithm iteratively obtains the values of $c_t$ using  Lemmas~\ref{lem:superset} and~\ref{lem:initial}. 
	Let $s=(t,I)$ be a sampling phase in $\mathcal{S}$.
	If $I$ consists of a single coordinate-index, say $i$, then $(u_t)_i$ was updated using Lemma~\ref{lem:superset} during the phase $s$. 
	Thus we have 
	\begin{equation*}
	\cost_I(R^s, c_t) \leq (1+\alpha)\cost_I(R^s, c(R^s)) \leq (1+\alpha)\cost_I(P_t^*,  c_t^*), 
	\end{equation*} 
	with probability $q$. 
	The first inequality holds by Lemma~\ref{lem:superset} with probability $q$.
	The second inequality holds because $R^s$ is a subset of $P_t^*$, and $c_t$ is
	the  centroid of $P_t^*$. 
	
	Otherwise, that is, if $I$ consists of more than one coordinate-indices, $(u_t)_i$'s were obtained using Lemma~\ref{lem:initial}
	for all indices $i\in I$ 	during the sampling phase $s$.
	Thus we have the following with probability $p$ by Lemma~\ref{lem:initial},
	\begin{equation*}
	\cost_I(R^s, c_t) \leq (1+\alpha)\cost_I(R^s, c(R^s)) \leq (1+\alpha)\cost_I(P_t^*,  c_t^*).
	\end{equation*} 
	
	By Lemma~\ref{inv: invariants}, the number of indices obtained by Lemma~\ref{lem:superset} is at most
	$2\Delta$ in total, and the number of indices obtained by Lemma~\ref{lem:initial} is exactly two, one for each cluster. 
	Therefore, with probability $p^2q^{2\Delta}$, we have
	\[
	\sum_{\substack{s\in\mathcal{S}}} \cost_{I^s}(R^s,c_{t^s})\leq (1+\alpha)\opt_2(P).\tag*{\qedhere}	
	\] 
\end{proof}

By combining Lemmas~~\ref{claim:wrong_assigned},
\ref{claim:correctly_assigned}, and~\ref{claim:alpha_opt}, we have the following lemma.
\begin{lemma}\label{lem:factor}
	For a constant $\alpha>0$, the algorithm returns an $(1+64c(\Delta+1)^2)(1+\alpha)$-approximate $2$-means clustering for $P$ with probability at least $p^2q^{2\Delta}$, where $q$ and $p$ are the probabilities in Lemmas~\ref{lem:superset} and~\ref{lem:initial}. 
\end{lemma}

\subsection{Algorithm without Counting Oracle}\label{sec:refined}
The algorithm described in Section~\ref{sec:two-algo} uses the counting oracle, which seems hard to implement. 
There are two places where the counting oracle is used: to determine the type of the phase and to determine the coordinate- and cluster-indices 
to be updated in a sampling phase (at Line~4-5 and at Line~10 of Algorithm~\ref{algo:2_means:ideal}). 

\begin{algorithm}[t]
	\SetAlgoLined
	$R \gets R-\fd(R,I_1\cap I_2)$\;
	$\mathcal{E}\gets \emptyset$, $u_1' \gets u_1$ and $u_2' \gets u_2$\;
	\lIf{$R=\emptyset$} {\KwRet{$(u_1,u_2)$}}
	
	\For{t=1,2}{
		\eIf{$I_t=\emptyset$}{
			$u_t' \gets$ the $\Delta$-missing point obtained from Lemma~\ref{lem:initial}\;
			Add the clustering returned by \textbf{\textsf{2-Means}}$((u_1',u_2'),R)$}		
		{\ForEach{$j\in [d]-I_t$}{
				$u_t' \gets u_t$\;
				$(u_t')_j \gets$ The value obtained from Lemma~\ref{lem:superset}\;
				Add the clustering returned by \textbf{\textsf{2-Means}}$((u_1',u_2'),R)$}
	}}
	
	$t \gets$ the cluster-index that maximizes $|\fd(R, I_t)|$\;
	\If{$\fd(R,I_t)\geq |R|/3$}{
		$B\gets$ The first half of $\fd(R,I_t)$ sorted in ascending order of distance from $u_t$\;
		Add the clustering returned by \textbf{\textsf{2-Means}}$((u_1,u_2),R-B)$ to $\mathcal{E}$\;
	}
	\KwRet{the clustering $(c_1,c_2)$ in $\mathcal{E}$ which minimizes $\cost(R,\{c_1,c_2\})$}
	
	\caption{\textbf{\textsf{2-Means}}$((u_1,u_2), R)$}\label{algo:2_means}
\end{algorithm}

In this section, we show how to avoid using the counting oracle. To do this, we simply try all possibilities: 
run both sampling and pruning phases, and update each of the indices during a sampling phase. 
Our main algorithm, \textsf{$2$-Means($(u_1,u_2), R$)}, is described in Algorithm~\ref{algo:2_means}. 
Its input consists of cluster centers $(u_1,u_2)$ of a partial clustering of $P$ and a set $R$ of points of $P$ which are not yet assigned to any cluster.
Finally, \mbox{\textsf{$2$-Means($(\otimes,\otimes), P$)}} returns an $(1+64c(\Delta+1)^2)(1+\alpha)$-approximate $2$-means clustering of $P$. By setting $\alpha$ property, we can obtain an $(1+\epsilon)$-approximate 2-means clustering of $P$.

\subparagraph{Description of the algorithm.}
In the main algorithm, we run both sampling and pruning phases. 
In a sampling phase, if $I_t\neq\emptyset$, we simply update $(u_t)_j$ for all coordinate-indices $j\in [d]-I_t$
since we do not know which index maximizes $|\pd(R\cap P_t^*, j)|$.  
At the end of sampling and pruning phases, we call  2-\mmeans recursively. 
Then we return the clustering with minimum cost among the clusterings returned the recursive calls. 
The pseudocode of this algorithm is described in Algorithm~\ref{algo:2_means}.

\subparagraph{Analysis of the algorithm.}
Clearly, the clustering cost returned by \textsf{$2$-Means($(\otimes,\otimes), P$)} is 
at most the cost returned by the algorithm described in Section~\ref{sec:two-algo}. 
In the following, we anlayze the running time of \textsf{$2$-Means($(\otimes,\otimes), P$)}.
Let $T(n,\delta)$ be the running time of \mbox{2-\means$((u_1,u_2), R)$}
when $n=|R|$ and $\delta=\min \{d-|I_1|,\Delta+1\}+\min \{d-|I_2|,\Delta+1\}$. 
Here, $\delta$ is an upper bound on the number of updates required to make $I_1=[d]$ and $I_2=[d]$. 

\begin{lemma}\label{claim:ind} $T(n,\delta)\leq \delta \cdot T(n,\delta-1)+T(\frac{5}{6}n,\delta)+ O(\frac{\delta \Delta^3 dn}{\alpha})$.
\end{lemma}

\begin{proof}
	Assume that we call \textsf{$2$-Means$((u_1,u_2),R)$} with $|R|=n$ and 
	$\delta=\min \{d-|I_1|,\Delta+1\}+\min \{d-|I_2|,\Delta+1\}$. 
	The tasks excluding Lines~6--17 and Line~23 take $O(\delta \Delta^3 dn/\alpha)$ time. In the following,
	we focus on the running time for completing the tasks described in Lines~6--17 and Line~20.
	
	We can perform Line~9 in $O(\Delta^3 dn/\alpha)$ time by Lemma~\ref{lem:initial}, and 
	Line~14 in $O(dn/\alpha)$ time by Lemma~\ref{lem:superset} 
	for each index $[d]-I_t$. 
	Since $u_t$ is updated at most $\delta$ times, this takes $O(\delta \Delta^3 dn/\alpha)$ in total. 
	Then 
	we recursively call 2-\means$((u_1',u_2'), R)$ at most $\delta$ times at Lines~6--17. 
	For each call to \mbox{2-\means$((u_1',u_2'), R)$}, $\delta$ decreases by one. 
	Therefore, the time for the recursive calls at Lines~6--17 is at most $\delta\cdot T(n,\delta-1)$. 
	
	Also, 2-\mmeans calls to itself recursively at Line~21. 
	In this case, the second parameter, $R-B$, has size at most $5n/6$. Therefore, the time for this recursive call is
	$T(5n/6, \delta)$. Therefore, the lemma holds.
\end{proof}

\begin{lemma}\label{claim:final} $T(n,\delta)\leq (6\delta)^{2\delta+1}{(\frac{6}{5})^{\delta^2}}\Delta^3 dn/\alpha$
\end{lemma}

\begin{proof}
	We prove the claim inductively. Basically, $T(n,0)\leq O(dn/\alpha)$ and $T(1,\delta)=O(1)$. For $\delta\geq 1$ and $n>1$, $T(n,\delta)$ satisfies the following by inductive hypothesis; 
	\begin{align*}
	T(n,\delta) & \leq \delta T(n,\delta-1)+T(\frac{5}{6}n,\delta)+\delta\cdot \Delta^3 dn/\alpha \\
	& \leq (6\delta)^{2\delta+1}{\left(\frac{6}{5}\right)^{\delta^2}}\Delta^3 dn/\alpha 
	\left(\left(\frac 5 6\right)^{2\delta-1}\left(\frac 1 6\right)   +\frac 5 6+\left(\frac1 6\right)^{2\delta}   \right)\\
	& \leq (6\delta)^{2\delta+1}{\left(\frac{6}{5}\right)^{\delta^2}}\Delta^3 dn/\alpha. \tag*{\qedhere}
	\end{align*}
\end{proof}

By Lemmas~\ref{lem:factor} and~\ref{claim:final}, we get the following lemma.
\begin{lemma}
	2-\mmeans$((\otimes,\otimes),P)$ returns a clustering for $P$ of cost at most $(1+64c(\Delta+1)^2)(1+\alpha)\opt_2(P)$ in $2^{O(\Delta^2-\log \alpha)}d|P|$ time with probability $p^2q^{2\Delta}$.
\end{lemma}

We obtain the following theorem by setting $\alpha=\epsilon/3$ and $c=\frac\alpha{64(\Delta+1)^2}$.
\begin{theorem}
	Given a $\Delta$-missing $n$-point set $P$ in $\mathbb{H}^d$, a $(1+\epsilon)$-approximate $2$-means clustering can be found in $2^{O(\max\{\Delta^4\log\Delta,\ \frac {\Delta }{\epsilon} \log{\frac 1 \epsilon}\})}dn$ time with probability $1/2$.
\end{theorem}

\section{For $k$-Means Clustering}\label{sec:k}
In this section, we describe and analyze for a $k$-means clustering algorithm. We also describe the algorithm for $k=2$ in Appendix~\ref{sec:for_2_means_clustering} to help readers to follow this section.
Let $\mathcal U = \langle u_1, u_2,\ldots, u_k\rangle$ be a sequence of points in $\mathbb{H}^d$.
For a sequence $T$ of cluster-indices, we use ${\mathcal U}_T$ to denote a $|T|$-tuple consisting of $u_t$'s for $t\in T$. 
In this section, for any sequence $\mathcal A$ and a set $T$ of sequential-indices, $\mathcal A_T$ denotes $|T|$-subsequence of $\mathcal A$ consisting of the $t$-th entries of $\mathcal A$ for $t\in T$.  

\subsection{Algorithm Using the Counting Oracle}\label{sub:algorithm_using_the_counting_oracle}
We first sketch an algorithm for $k$-clustering assuming that we can access the \emph{counting oracle}. Let $\mathcal P^*=\langle P_1^*, \cdots, P_k^* \rangle$ be an optimal $k$-clustering for $P$ induced by the centroids $\mathcal C^* =\langle c_1^*, \cdots, c_k^* \rangle$. The counting oracle takes a subset of $P$ and a cluster-index $t\in [k]$, and it returns the number of points in the subset which are contained in $P_t^*$. Then in Section~\ref{sec:k-without-oracle},
we show how to do this without using the counting oracle. 

The algorithm consists of several phases of two types: a \emph{sampling phase} or a \emph{pruning phase}. 
We initialize $\mathcal U$ so that $(u_t)_i=\otimes$ for all $t\in [k]$ and $i\in[d]$.
In a sampling phase, we obtain values of $(u_t)_j$ for indices $t\in [k]$ and $j\in[d]$ which were set to $\otimes$.  
Also, we assign points of $P$ to one of the $k$ clusters in sampling and pruning phases. 
Finally, $\mathcal U$ is updated to a $(1+\epsilon)$-approximate clustering as we will see later.
In the following, we use $\mathcal C$ to denote the output of the algorithm. 
The pseudocode of the algorithm is described in Algorithm~\ref{algo:k_means:ideal}.

\begin{algorithm}[t]
	\SetAlgoLined
	\Input{A set $P$ of $\Delta$-missing points in the plane}
	\Output{A $(1+\epsilon)$-approximate $k$-means clustering for $P$}
	
	\SetNoFillComment 
	
	$R \gets P$ and $P_t\gets\emptyset$ for all cluster-indices $t\in [k]$\;
	Initialize $\mathcal U=\langle u_1,\dots, u_k \rangle$ so that $(u_t)_i=\otimes$ for all cluster-indices $t\in [k]$ and $i\in[d]$\;
	\While{$R\neq\emptyset$} {	
		Let $t$ be the cluster-index that maximizes $|\pd(R\cap P_t^*,[d]-I_t)|$\;
		
		\If{$|\pd(R\cap P_t^*,[d]-I_t)|\geq c|R|$}{
			\tcc*[l]{sampling phase}
			{\If{$I_t=\emptyset$}
				{
					$u_t \gets$ The $\Delta$-missing point obtained from Lemma~\ref{lem:initial}\;
				}
				\Else
				{	Let $j$ be the coordinate-index in $[d]-I_t$ that maximizes $\pd(R\cap P_t^*, j)$\;
					$(u_t)_j \gets$ The value obtained from Lemma~\ref{lem:initial}\;
				}
				Assign the points in $\fd(R,\cap_{t\in [k]} I_t)$ to their closest cluster centers in $u_{[k]}$\;	
				$R \gets R-\fd(R,\cap_{t\in[k]}I_t)$\;	
		}}
		
		\Else{
			\tcc*[l]{pruning phase}
			$T \gets$ The set of cluster-indices that maximizes $|S_T|$\;
			$B\gets$ The first half of $S_T$ sorted in ascending order of distance from $\mathcal U_T$\;
			Assign the points in $B$ to their closest cluster centers in $\mathcal U_T$\;
			$R \gets R-B$\;
		}	
	}
	\KwRet{$\mathcal U$} 
	\caption{\textbf{\textsf{Idealized $k$-Means}}}\label{algo:k_means:ideal}
\end{algorithm}

At any moment during the execution of the algorithm, we maintain the set $R$ of remaining points 
and a $k$-tuple $\mathcal U=\langle u_1,\ldots, u_k \rangle$ of (partial) centers in $\mathbb{H}^d$. 
Also, we let $I_t$ be $\dom(u_t)$. 
Initially, $R$ is set to $P$, and $\mathcal U$ is set to the $k$-tuple of null points. 
The algorithm terminates if $R=\emptyset$, and finally $\mathcal U$ becomes a set of points in $\mathbb{R}^d$. 
At the beginning of each phase,  we decide the type of the current phase. 
For this purpose, we consider the partition $\mathcal{F}$ of $R$ defined as follows.

For a subset $T$ of $[k]$, let $S_T$ denote 
the set of points $x\in R$ such that 
$\dom(x) \subset I_t$ for every $t\in T$ and $\dom(x) \not\subset I_{t'}$ for every $t'\notin T$. 
Let $\mathcal{F}=\{S_T\mid T \textnormal{ is a proper subset of }[k]\}$.
The following lemma shows that $\mathcal F$ is a partition of $R$. 

\begin{lemma}\label{lem:partitioning}
	For any point $x\in R$, there exists a unique set in $\mathcal{F}$
	containing $x$. 
\end{lemma}
\begin{proof}
	We first show that the sets in $\mathcal{F}$ are pairwise disjoint. 
	Consider two proper subsets $T$ and $T'$ of $[k]$. Then we have $S_T\cap S_{T'}=\emptyset$. 
	To see this, let $t$ be a cluster-index in $T-T'$. 
	By definition, for a point $p\in S_T$, we have $\dom(p)\subset I_t$, and 
	for a point $p \in S_{T'}$, we have $\dom(p)\not\subset I_t$. 
	
	Thus it suffices to show that the union of the sets in $\mathcal{F}$ is $R$. For any point $p\in R$, 
	let $T$ be the set of cluster-indices $t$ in $[k]$ with $\dom(p)\subset I_t$.\footnote{Notice that $T$ might be empty, but it is also a proper subset of $[k]$.}
	By definition, $p$ is contained in $S_T$, 
	which is a set of $\mathcal{F}$. Therefore, the lemma holds. 
\end{proof}

At the beginning of each phase, we decide which type of the current phase. 
Let $t$ be the cluster-index of $[k]$ that maximizes $|\pd(R\cap P_{t}^*,[d]-I_{t})|$, where $R$ is the set of points of $P$ which are not assinged to any cluster. The one of the following cases always happen: 
$|\pd(R\cap P_{t}^*,[d]-I_{t})|\geq c|R|$, or
there exists a set $T$ of cluster-indices in $[k]$ such that  $|S_{T}\cap(\cup_{t\in  T} P_t^*)|\geq c|R|$ for any constant $c<1/(2^k+k)$, which will be specified later.\footnote{We set $\alpha=\epsilon/3$, and $c=\frac\alpha{8\cdot2^k k^2(\Delta+1)^2}$.} 

\begin{lemma}\label{lem:partitioning_T}
	One of the following always holds for any constant $c<1/(2^k+k)$.
	\begin{itemize}
		\item $|\pd(R\cap P_{t}^*,[d]-I_{t})|\geq c|R|$ for a cluster-index $t\in[k]$, or 
		\item $|S_{T}\cap(\cup_{t\in  T} P_t^*)|\geq c|R|$ for a proper subset $T$ of $[k]$.
	\end{itemize}
\end{lemma}
\begin{proof}
	By the definition of $\mathcal F$, we have 
	\[
	\pd(R\cap P_t^*,[d]-I_t)=\bigcup_{t\notin T\subsetneq [k]} S_T\cap P_t^*.
	\]
	Thus, the union of $\{\pd(R\cap P_t^*,[d]-I_t)\}_{t\in [k]}$ $\{S_T\cap (\cup_{t\in T} P_t^*)\}_{T\subsetneq [k]}$ decompose $R$. Therefore, the lemma holds.
\end{proof}

\subparagraph{Sampling phase.} If the first case happens, we enter a sampling phase.
Let $\alpha$ be a constant, which will be specified later.
We use it as an approximation factor used in Lemmas~\ref{lem:superset} and~\ref{lem:initial} for sampling.  
If $I_t$ is empty, we replace $u_t$ with a $\Delta$-missing point in $\mathbb{H}^d$  obtained from Lemma~\ref{lem:initial}.  
If $I_t$ is not empty, then it is guaranteed that $|I_t|$ is at least $d-\Delta$. 
We compute the coordinate-index $j$ in $[d]-I_t$ that maximizes $|\pd(R\cap P_t^*,j)|$
using the counting oracle. Clearly, $(u_t)_j=\otimes$ and 
$|\pd(R\cap P_t^*,j)|$ is at least $c|R|/\Delta$.
Then we replace $(u_t)_j$ with a value obtained from Lemma~\ref{lem:superset}.
At the end of the phase, we check if $\fd(R, \cap_{t\in [k]} I_t)$ is not empty.
If it is not empty, we assign those points to their closest cluster centers. 

\subparagraph{Pruning phase.}

Otherwise, we enter a pruning phase. Instead of obtaining a new coordinate value of $u_t$, 
we assign points of $R$ to cluster centers in a pruning phase. 
To do this, we find a proper subset $T$ of $[k]$ which maximizes $| S_{ T}|$. Then  among the points of $S_T$, we choose the $|S_T|/2$ points closest to their closest centers in $u_T$.
Then we assign each of them to its closest center in $\mathcal U_T$. 
In this way, points in $\cup_{t'\notin T} P_{t'}^*$ might be assigned (incorrectly) to $u_t$ for a cluster-index $t\in T$.
We call such a point a \emph{stray point}.

\subsection{Analysis of the Approximation Factor}


In this section, we analyze the approximation factor of the algorithm. 
We let 
$\mathcal{S}$ be the sequence of sampling phases happened during the execution of the algorithm.

\begin{lemma} 
	At any time during the execution,  
	$|I_t|\geq d-\Delta$ or $|I_t|=0$ for $t\in[k]$.
\end{lemma}
\begin{proof}
	Initially, all coordinates of $u_t$ are set to $\otimes$, and thus $I_t=\emptyset$.
	The cluster centers $u_t$ are updated during the sampling phases only.
	If $I_t$ is empty for $t\in[k]$, we use Lemma~\ref{lem:initial}, and thus $u_t$ is updated to
	a $\Delta$-missing point. 
	If $I_t$ is not empty, we use Lemma~\ref{lem:superset}, and increase $|I_t|$ by one.
	Therefore, the lemma holds. 
	%
	%
\end{proof}

\begin{corollary}\label{k:inv}
The size of $\mathcal{S}$ is at most $k(\Delta+1)$.
\end{corollary}

\subsubsection{Clustering Cost Induced by Stray Points}
We first analyze the clustering cost induced by the stray points assigned during the pruning phases.
Consider the consecutive pruning phases lying between two adjacent sampling phases in $\mathcal {S}$. 
Let $N$ denote the number of such pruning phases.  
During this period, $u_t$ remains the same for each cluster-index $t\in[k]$.
For a proper subset $T$ of $[k]$, let $A_T^{(x)}$ denote the set of points of $S_T$ assigned to $u_t$ for an cluster-index $t\in T$ at the $x$th iteration (during this period). 
Let $R^{(x)}$ be the set of points of $P$ which are not yet assigned to any cluster
	at the \emph{beginning} of the $x$th iteration (during this period). 
Let $X_T^{(x)}=R^{(x)}\cap S_T$. 
By construction, in the $x$th iteration, there exists a unique index-set $T$
with $A_T^{(x)}\neq\emptyset$. 
Furthermore, let $\mathcal X_T$ be the increasing sequence of indices $x$ of $[N]$ with $A_T^{(x)}\neq \emptyset$. For convenience, let $A_T^{(N+1)}=X_T^{(N)}$.  We denote $P_T^*=\cup_{t\in T} P_t^*$ and $P_{\bar T}^*=\cup_{t'\notin T} P_{t'}^*$.
We need the following three technical claims, which will be used to prove Lemma~\ref{claim:k_wrong_assigned}. 


\begin{claim}\label{claim:k-base} For any index $x\in[N]$, we have the following.
\[
\sum_{T \subsetneq [k]} |X_T^{(x)}\cap P_{\bar T}^*| < ck(\Delta+1) \sum_{T \subsetneq [k]} |X_T^{(x)}|.
\]
\end{claim}

\begin{proof}
	Assume to the contrary that $\sum |X_T^{(x)}\cap P_{T'}^*| \geq ck(\Delta+1) \sum |X_T^{(x)}|$.
	In the following, we show that there is a cluster-index $t$ such that $|\pd(P^*_t\cap R^{(x)},[d]-I_t)|\geq c|R^{(x)}|$.
	This contradicts that that the $x$th phase is a prunning phase. Let $R'=R^{(0)}-\cup X_T^{(0)}$.
	That is, $R'$ is the set of points which are not yet assigned to any cluster at the beginning of these pruning phases, but whose domains are not contained in $S_T$ for any proper subset of $[k]$.  
	Therefore, during the pruning phases (in this period), no point of $R'$ is assigned to any cluster.
\begin{align*}
\sum_{t\in [k]} |\pd(P_{t}^*\cap R^{(x)},[d]-I_{t})|
&=\sum_{t\in [k]} |\pd(P_{t}^*\cap (R'\cup (\cup_{T } X_{T}^{(x)})),[d]-I_{t})|\\
&=\sum_{t\in [k]} \sum_{t\notin T \subsetneq [k]} |P_{t}^*\cap X_T^{(x)}| + \sum_{t\in [k]}|P_{t}^{*}\cap R'| \\
&=\sum_{T\subsetneq [k]} \sum_{t\notin T} |P_{t}^*\cap X_T^{(x)}|+\sum_{t\in [k]}|P_{t}^{*}\cap R'| \\
&=\sum_{T\subsetneq [k]} |P_{\bar T}^*\cap X_T^{(x)}|+|R'|\\
&\geq ck(\Delta+1) (\sum |X_T^{(x)}|+|R'|)\\
&= ck(\Delta+1)|R^{(x)}|.
\end{align*}
The first and the last inequality hold since $R'$ and $X_T^{(x)}$ partition $R^{(x)}$. The second equality holds since $X_T^{(x)}\subset S_T$.
Also, $\dom(p)\subset I_t$ for any point $p\in S_T$ if and only if $t\in T$. The equalities hold since they only change the ordering of summation and the inequality holds by assumption and the fact $ck(\Delta+1)\leq 1$. 
Therefore, the claim holds.
\end{proof}


\begin{claim}\label{claim:k-frac}
For any consecutive indices $x$ and $x'$ in $\mathcal X_T$ with $x'<x$, the following holds for any proper subset $T$ of $[k]$: 
\[\frac{|A_T^{(x')}\cap P_{\bar T}^*|}{|A_T^{(x)}\cap P_T^*|}\leq \frac{4\cdot 2^kck(\Delta+1)}{1-2\cdot2^kck(\Delta+1)}.\]
\end{claim}
\begin{proof}
		We first show that $|A_T^{(x)}\cap P_{\bar T}^*|$ is at most $4c(2\Delta+2){|X_1^{(x)}|}$ as follows.
		
		\begin{align*}
		|A_T^{(x)}\cap P_{\bar T}^*|&\leq |X_T^{(x')}\cap P_{\bar T}^*|\\
		&<ck(\Delta+1)(\sum_{T' \subsetneq [k]} |X_{T'}^{(x')}|) \\
		&\leq 2^kck(\Delta+1)|X_T^{(x')}|\\
		&\leq 2\cdot 2^kck(\Delta+1)|A_T^{(x)}|.
		\end{align*}
		The first inequality holds because $A_T^{(x)}$ is a subset of $X_T^{(x')}$.
		The second inequality holds by Claim~\ref{claim:k-base}, and the third inequality holds
		since $T$ maximizes $|X_T^{(x')}|$. If it is not the case, then $A_T^{(x)}=\emptyset$ since the algorithm would assign the first half of $X_{T'}^{(x')}$ for a set $T'\subsetneq [k]$ other than $T$.  This contradicts for $x\in \mathcal{X}_T$. The last inequality holds since $|X_T^{(x')}|\leq2|A_T^{(x)}|$.
		
		Now we give a lower bound of $|A_T^{(x)}\cap P_T^*|$. The second inequality holds due to the upper bound of $|A_T^{(x)}\cap P_{\bar T}^*|$ stated above. 
		\begin{align*}
		|A_T^{(x)}\cap P_T^*|&=|A_T^{(x)}|-|A_T^{(x)}\cap P_{\bar T}^*|\\
		&\geq 	|A_T^{(x)}|- {2\cdot 2^kck(\Delta+1) |A_T^{(x)}|}\\
		&=	(1-2\cdot 2^kck(\Delta+1)){|A_T^{(x)}|}
		\end{align*}
		By combining the bounds, we can obtain
		$\frac{|A_T^{(x')}\cap P_{\bar T}^*|}{|A_T^{(x)} \cap P_T^*|} \leq \frac{2\cdot 2^kck(\Delta+1)|A_T^{(x')}|}{(1-2\cdot 2^kck(\Delta+1))|A_T^{(x)}|}\leq \frac{4\cdot 2^kck(\Delta+1)}{1-2\cdot 2^kck(\Delta+1)}$.
		The last inequality holds by $|A_T^{(x')}|\leq2|A_T^{(x)}|$.
		Therefore, the claim holds.
	\end{proof}


\begin{claim}\label{cliam:k-wrong_base}
Let $A_T = \cup_{x\in[N]} A_T^{(x)}$. 
For any proper subset $T$ of $[k]$, the following holds:
\[
	\cost(A_T\cap P_{\bar T}^*,\mathcal C_T)\leq 8\cdot 2^kck(\Delta+1)\cost(X_T^{(0)}\cap P_T^*,\mathcal C_T).
\]
\end{claim}
\begin{proof}
		Note that each $c_t$ of $\mathcal C$ is extension of $u_t$. 
		Recall that $\mathcal{X}_T$ is the increasing sequence of indices $x$ of $[N+1]$
		with $A_T^{(x)}\neq\emptyset$ and the last index is $N+1$. 
		By definition, 
		for any index $x\in[N]$, $\mathcal C_T$ is closer to any point of $A_T^{(x)}$ than
		to any point of $X_T^{(x)}$. 
			Consider two consecutive indices $x$ and $x'$ in $\mathcal{X}_T$ with $x'<x$. 
		\[
		\frac{\cost(A_T^{(x')}\cap P_{\bar T}^*, \mathcal C_T)}{|A_T^{(x')}\cap P_{\bar T}^*|}\leq\frac{\cost(A_T^{(x)}\cap P_T^*,\mathcal C_T)}{|A_T^{(x)}\cap P_T^*|}.\] 
		By Claim~\ref{claim:k-frac}, we have
		\[\cost(A_T^{(x')}\cap P_{\bar T}^*, \mathcal C_T)  \leq\frac{4\cdot 2^kck(\Delta+1)}{1-2\cdot2^kck(\Delta+1)}\cost(A_T^{(x)}\cap P_T^*,\mathcal C_T).\]
		
		Note that the points of $P$ assigned to $u_t$, where $t\in T$, during these consecutive pruning phases are 
		exactly the points in the union of $A_T^{(x)}$'s for all indices $x\in[N]$. 
		Also, $\{A_T^{(x)}\}_{x \in\mathcal{X}_T}$  is pairwise disjoint
		and $\frac{4\cdot 2^kck(\Delta+1)}{1-2\cdot2^kck(\Delta+1)}\leq 8\cdot 2^kck(\Delta+1)$. Thus, the claim holds.
\end{proof}
\medskip

Recall that $\mathcal{S}$ denotes the sequence of sampling phases. 
Here, we represent a sampling phase as the pair $(t, I)$,
where $t$ is the cluster-index considered in the sampling phase, and $I$ is the set of indices $i$
such that $(u_t)_i$ is obtained during the sampling phase. 
Also, for $s=(t,I)\in\mathcal{S}$, we let $t^s=t$ and $I^s=I$.
For
 two sampling phases $s$ and $s'$ in $\mathcal{S}$,
 we use $s\preceq s'$ if $s$ comes before $s'$ in $\mathcal{S}$ or equals to $s'$.

 For each sampling phase $s=(t,I)$, let $R^s$ be the set of points of $P_t^*$ which are not assigned at the \emph{beginning} of the phase. Furthermore, let $I_{t'}^s$ denote $\dom(u_{t'})$ we have at the end of the sampling phase $s$
for a cluster-index $t'\in[k]$. 
	  For a proper subset $T$ of $[k]$ and a sampling phase $s\in \mathcal{S}$, let $A_T^s$ denote the set of points in $S_T$ assigned to $P_T$ during  the pruning phases lying between $s$ and $s'$,
	  where $s'$ is the sampling phase coming after $s$ in $\mathcal{S}$. 

\medskip
\begin{claim}\label{claim:fd2pd}
For a sampling phase $s'$ in $\mathcal{S}$ and a proper subset $T$ of $[k]$, let $X$ be a point subset of $S_T$ which are not assigned at the end of a sampling phase $s'$. Then we have,
\[
	\cost(X\cap P_T^*,\mathcal C_T)\leq \sum_{\substack{s}} \cost_{I^s}(\pd(X\cap P_{t^{s}}^*,I^s),c_{t^s}),
\] where the summation is taken over all sampling phases $s$ in $\mathcal{S}$ which $s\preceq s'$ and $t^s\in T$.
\end{claim}
\begin{proof} 
	Note that all points of $X\subset S_T$ are fully defined on $I_t$ for every cluster-indices $t\in T$ during the pruning phases between $s'$ and its consecutive sampling phase. Also, $I_t^{s}\subset I_t^{s'}$ for sampling phases $s$ with $s\preceq s'$. Then, by definition of $\cost(\cdot)$, we have the following, 
	\begin{align*}
	\cost(X\cap P_T^*,c_T)
	&\leq \sum_{t\in T}\cost(X\cap P_t^*,c_t)\\
	&=\sum_{t\in T}\cost_{I_t^{s'}}(\fd(X\cap P_t^*,I_t^{s'}),c_t)\\
	&\leq \sum_{\substack{s\preceq s'\\ t^{s}\in T}} \cost_{I^s}(\pd(X\cap P_{t^{s}}^*,I^s),c_{t^s}).\tag*{\qedhere}
	\end{align*}
\end{proof}

The following lemma gives an upper bound of the cost induced by the stray points. 
\begin{lemma}\label{claim:k_wrong_assigned} $\displaystyle{\sum_{s\in \mathcal{S}} \sum_{T\subsetneq [k]}\cost(A_T^s\cap P_{\bar T}^*,\mathcal C_T) \leq 8\cdot 2^k ck^2(\Delta+1)^2\sum_{\substack{s\in\mathcal{S}}} \cost_{I^s}(R^s,c_{t^s})}.$
\end{lemma}
\begin{proof}
	
	Let $X_T^s$ be the set of points in $S_T$ which are not assigned at the end of a sampling phase $s$.
	Then, we have the following;
	\begin{align*}
	\sum_{s\in \mathcal{S}} \sum_{T\subsetneq [k]}\cost(A_T^s\cap P_{\bar T}^*,\mathcal C_T)
	&\leq 8\cdot 2^kck(\Delta+1)\sum_{s\in \mathcal{S}}\sum_{T\subsetneq [k]}  \cost(X_T^s\cap P_T^*,\mathcal C_T)\\
	&\leq 8\cdot 2^kck(\Delta+1)\sum_{s'\in \mathcal{S}}\sum_{T\subsetneq [k]}\sum_{\substack{s\preceq s' \\ t^s\in T}} 
	\cost_{I^{s}}(\pd(X_T^{s'}\cap P_{t^{s}}^*,I^{s}),c_{t^{s}})\\
	&= 8\cdot 2^kck(\Delta+1)\sum_{s \in S}\sum_{\substack{s\preceq s'}}\sum_{t^s\in T\subsetneq [k]} \cost_{I^{s}}(\pd(X_T^{s'}\cap P_{t^{s}}^*,I^{s}),c_{t^{s}})\\
	&\leq 8\cdot 2^kck(\Delta+1)\sum_{\substack{s\in\mathcal{S}}}\sum_{s\preceq s'} \cost_{I^{s}}(\pd(R^{s},{I^{s}}),c_{t^{s}})\\
	&\leq 8\cdot 2^kck^2(\Delta+1)^2\sum_{\substack{s\in\mathcal{S}}} \cost_{I^{s}}(\pd(R^{s},{I^{s}}),c_{t^{s}})\\
	&\leq  8\cdot 2^k ck^2(\Delta+1)^2\sum_{\substack{s\in\mathcal{S}}} \cost_{I^s}(R^s,c_{t^s}), 
	\end{align*}
	The first and second inequalities hold by Claims~\ref{cliam:k-wrong_base} and~\ref{claim:fd2pd}. 
	The third one holds since it changes only the ordering of summation. 
	The fourth one holds since for a fixed sampling phase $s'$ in $\mathcal{S}$, $X_T^{s'}$ are disjoint for all proper subsets $T$ of $[k]$.
	Also notice that, for two sampling phases $s,s'$ in $\mathcal{S}$ and a proper subset $T$ of $[k]$, $R^s$ contains $X_T^{s'}\cap P_{t^{s}}^*$ if $s\preceq s'$.
	 The fifth one holds since the size of $\mathcal{S}$ is at most $k(\Delta+1)$. The last holds by the definition of $\cost(\cdot)$. 
\end{proof}

\subsubsection{Clustering Cost Induced by Non-Stray Points}
We then give an upper bound on the clustering cost induced by non-stray points.
The first term in the following lemma is the clustering cost induced by points assigned during the sampling phases,
and the second term is the clustering cost induced by non-stray points assigned during the pruning phases. 

\begin{lemma}\label{claim:k_correct_assigned}
Let $S$ be the set of points assigned during the sampling phases. 
\[
	\cost(S,\mathcal C)+\sum_{s\in \mathcal{S}}\sum_{T\subsetneq [k]}\cost(A_T^s\cap P_T^*,\mathcal C_T)\leq \sum_{\substack{s\in\mathcal{S}}} \cost_{I^s}(R^s,c_{t^s}).
\]

\end{lemma}
\begin{proof}
During the sampling phases, we assign points in $S$ according to the Voronoi partition of $u_{[k]}$. Thus, we have the following.
The last inequality holds since we assign a point during a sampling phase if it is fully defined on $\dom(u_t)$ for all cluster-indices $t\in [k]$.
\begin{align*}
\cost(S,\mathcal C)\leq \sum_{t\in [k]}\cost(S\cap P_t^*,c_t)=\sum_{\substack{s\in \mathcal{S}}} \cost_{I^s}(\pd(S\cap R^s,I),c_{t^s}).
\end{align*}

For the second term of this claim, we have the following; 
\begin{align*}
	\sum_{s'\in \mathcal{S}}\sum_{T\subsetneq [k]}\cost(A_T^{s'}\cap P_T^*,c_T)
	&\leq \sum_{s'\in \mathcal{S}}\sum_{T\subsetneq [k]}\sum_{\substack{s\preceq s' \\ t^s\in T}} \cost_{I^{s}}(\pd(A_T^{s'}\cap P_{t^{s}}^*,I^{s}),c_{t^{s}})\\ 
	&=\sum_{s\in S}\sum_{\substack{s\preceq s'}}\sum_{t^s\in T\subsetneq [k]} \cost_{I^{s}}(\pd(A_T^{s'}\cap P_{t^{s}}^*,I^{s}),c_{t^{s}})\\
	&\leq \sum_{s\in S}\sum_{\substack{s\preceq s'}} \sum_{t^s\in T\subsetneq [k]}\cost_{I^s}(\pd(A_T^{s'}\cap R^{s},I^{s}),c_{t^s})\\
	&\leq \sum_{\substack{s\in S} }\cost_{I^s}(\pd(A \cap R^{s},I^s),c_{t^s}),
\end{align*}
where $A$ denotes the set of points of $P$ assigned during the pruning phases.
The first equality holds by Claim~\ref{claim:fd2pd}.
 The second and the last inequalities hold since they change only the ordering of summation. The third equality holds since $A_T^{s'}\cap P_{t^s}^*\subset R^{s}$ if $s\preceq s'$. 

By combining previous properties, we have;
\begin{align*}
	\cost(S,\mathcal C)+\sum_{s\in \mathcal{S}}\sum_{T\subsetneq [k]}\cost(A_T^s\cap P_T^*,\mathcal C_T)
	&\leq \sum_{\substack{s\in \mathcal{S}}} \cost_{I^s}(\pd(S\cap R^s,I^s),c_{t^s})\\
	&+\sum_{\substack{s\in S} } \cost_{I^s}(\pd(A \cap R^{s},I^s),c_{t^s})\\
	&\leq \sum_{\substack{s\in\mathcal{S}}} \cost_{I^s}(R^s,c_{t^s}).\tag*{\qedhere}
\end{align*}
\end{proof}

Since the total clustering cost is bounded  by 
the sum of $\cost_I(R^s,c_t)$ for all sampling phases $s=(t,I)$ 
within a factor of $(1+8\cdot 2^k ck^2(\Delta+1)^2)$ by Lemma~\ref{claim:k_wrong_assigned} and~\ref{claim:k_correct_assigned},
it suffices to bound 
the sum of $\cost_I(R^s,c_t)$ for all sampling phases $s=(t,I)$.
The following lemma can be proved using Lemmas~\ref{lem:superset} and~\ref{lem:initial}. Its proof can be found in Appendix~\ref{apd:proof}.
\begin{lemma}\label{claim:alpha_k_opt} For a constant $\alpha>0$, $
\sum_{\substack{s\in\mathcal{S}}} \cost_{I^s}(R^s,c_{t^s})\leq (1+\alpha)\opt_k(P)
	$ with a probability at least $p^kq^{k\Delta}$, where $q$ and $p$ are the probabilities in Lemmas~\ref{lem:superset} and~\ref{lem:initial}.
\end{lemma}
\begin{proof}
	The algorithm iteratively obtains the values of $c_t$ using  Lemmas~\ref{lem:superset} and~\ref{lem:initial}. 
	For a sampling phase $s=(t,I)$ in $\mathcal{S}$, if $I$ consists of a single coordinate-index, say $i$, then $(u_t)_i$ was updated using Lemma~\ref{lem:superset}. 
	Thus we have 
	\begin{equation*}
	\cost_I(R^s, c_t) \leq (1+\alpha)\cost_I(R^s, c(R^s)) \leq (1+\alpha)\cost_I(P_t^*,  c_t^*).
	\end{equation*} 
	The first inequality holds by Lemma~\ref{lem:superset} with probability $q$.
	The second inequality holds because $R^s$ is a subset of $P_t^*$, and $c_t$ is
	the  centroid of $P_t^*$. 
	
	Otherwise, that is, if $I$ consists of more than one coordinate-indices, $(u_t)_i$'s were obtained using Lemma~\ref{lem:initial}
	for all indices $i\in I$.
	Thuw we have the following with a probability at least $p$ by Lemma~\ref{lem:initial},
	\begin{equation*}
	\cost_I(R^s, c_t) \leq (1+\alpha)\cost_I(R^s, c(R^s)) \leq (1+\alpha)\cost_I(P_t^*,  c_t^*).
	\end{equation*} 
	
	By Lemma~\ref{inv: invariants}, the number of indices obtained by Lemma~\ref{lem:superset} is at most
	$k\Delta$ in total, and the number of indices obtained by Lemma~\ref{lem:initial} is exactly $k$, one for each cluster. 
	Therefore, with a probability at least $p^kq^{k\Delta}$, we have
	\[
	\sum_{\substack{s\in\mathcal{S}}} \cost_{I^s}(R^s,c_{t^s})\leq (1+\alpha)\opt_k(P).\tag*{\qedhere}	
	\] 
\end{proof}

We can obtain the following lemma by combining previous properties.
\begin{lemma}\label{lem:k-factor}
		For a constant $\alpha>0$, the algorithm returns an $(1+8\cdot 2^kck^2(\Delta+1)^2)(1+\alpha)$-approximate $k$-means clustering for $P$ with probability at least $p^kq^{k\Delta}$, where $q$ and $p$ are the probabilities in Lemmas~\ref{lem:superset} and~\ref{lem:initial}. 
\end{lemma}

\subsection{Algorithm without Counting Oracle}\label{sec:k-without-oracle}
The algorithm we have described uses the counting oracle in two places: to determine the type of the phase and selecting a pair of the cluster-index and coordinate-index to be updated in a sampling phase. 
In this section, we explain how to avoid using the counting oracle. 
To do this, we simply try all possible cases: run both phases and update each possible cluster for all indices during a sampling phase. The main algorhm, $k$-\mmeans$(\mathcal U,R)$, is described in Algorithm~\ref{algo:k_means}. Its input consists of cluster centers $\mathcal U$ of a partial clustering of $P$ and a set $R$ of points of $P$ which are not yet assigned. 
Finally, $k$-\mmeans$(\otimes_{[k]},P)$ returns an $(1+8\cdot 2^kck^2(\Delta+1)^2)(1+\alpha)$-approximate $k$-means clustering of $P$, where $\otimes_{[k]}$ denotes the $k$-tuple of $\mathbb H^d$ which all elements has $\otimes$ only for its all $d$-cordinates. By setting $\alpha$, we can obtain an $(1+\epsilon)$-approximate.

\begin{algorithm}[t]
	\DontPrintSemicolon
	$R \gets R-\fd(R,\cap_{t\in [k]}I_t)$\;
	$\mathcal{E}\gets \emptyset$, $u_t' \gets u_t$ for all $t\in[k]$\;
	\lIf{$R=\emptyset$} {\KwRet{$u_{[k]}$}}
	
	\For{$t\in [k]$}{
		\eIf{$I_t=\emptyset$}{
			$u_t' \gets$ the $\Delta$-missing point obtained from Lemma~\ref{lem:initial}\;
			Add the clustering returned by \textbf{\textsf{$k$-Means}}$(\mathcal U',R)$  to $\mathcal{E}$}		
		{\ForEach{$j\in [d]-I_t$}{
				$u_t' \gets u_t$\;
				$(u_t')_j \gets$ The value obtained from Lemma~\ref{lem:superset}\;
				Add the clustering returned by \textbf{\textsf{$k$-Means}}$(\mathcal U',R)$ to $\mathcal{E}$}}}
	$T \gets$ the non-empty proper subset of $[k]$ that maximizes $|R\cap S_T|$\;
	\If{$|R\cap S_T|\geq |R|/(2^k-1)$}{
		$B\gets$ The first half of $|R\cap S_T|$ sorted in ascending order of distance from $u_T$\;
		Add the clustering returned by \textbf{\textsf{$k$-Means}}$(\mathcal U,R-B)$ to $\mathcal{E}$\;
	}
	\KwRet{the clustering $\mathcal C$ in $\mathcal{E}$ which minimizes $\cost(R,\mathcal C)$}
	
	\caption{\textbf{\textsf{$k$-Means}}$(\mathcal U, R)$}\label{algo:k_means}
\end{algorithm}

The clustering cost returned by $k$-\mmeans$(\otimes_{[k]},P)$ is at most the cost returned by the algorithm which uses the counting oracle in Section~\ref{sub:algorithm_using_the_counting_oracle}. In the following, we analyze the running time of $k$-\mmeans$(\otimes_{[k]},P)$. Let $T(n,\delta)$ be the running of  $k$-\mmeans$(\mathcal U,R)$ when $n=|R|$ and $\delta=\sum_{t\in [k]} \min \{d-|I_t|,\Delta+1\}$. 
Here, $\delta$ is an upper bound on the number of updates required to make $I_t=[d]$ for every cluster-index $t$ in $[k]$. 

\begin{claim} $T(n,\delta)\leq \delta\cdot T(n,\delta-1)+T\left(\left(1-\frac{1}{2^{k+1}-2}\right)n,\delta\right)+O(\frac{k\delta\Delta^3 dn}{\alpha})$ 
\end{claim}
\begin{proof}
	In a sampling phase, $k$-\mmeans calls itself at most $\delta$ times recursively with different parameters. Each recursive call takes  $T(n,\delta-1)$ time. Also, the time for updating cluster centers takes $O(\delta\Delta^3 dn/\alpha)$ in total by Lemma~\ref{lem:superset} and~\ref{lem:initial}.
	For a pruning phase, we compute $|R\cap S_T|$ for each $T\subsetneq [k]$ in total $O(dn)$ time, and then choose the first half of $S_T$ 
	in increasing order of the distances from $u_T$ in total $O(kdn)$ time.
	The recursive call invoked in the pruning phase takes $T\left(\left(1-\frac{1}{2^{k+1}-2}\right)n,\delta\right)$ time.
	We have $\delta+1$ clusterings returned by recursive calls in total, and we can choose
	$c_{[k]}$ in $O(\delta kdn)$ time. Thus, the claim holds.
\end{proof}

\begin{claim}\label{claim:final_general} $T(n,\delta)\leq (2\delta(2^k-1))^{2\delta+1}{(1+\frac{1}{2^{k+1}-3})^{\delta^2}}\Delta^3 kdn/\alpha$
\end{claim}
\begin{proof}
	We prove the claim inductively. Basically, $T(n,0)\leq O(kdn)$ and $T(1,\delta)=O(kd)$. 
	
	For $\delta\geq 1$ and $n>1$, $T(n,\delta)$ satisafies the following by inductive hypothesis;
	
	\begin{align*}
	T(n,\delta) & \leq \delta\cdot T(n,\delta-1)+T\left(\left(1-\frac{1}{2^{k+1}-2}\right)n,\delta\right)+\delta\cdot \Delta^3 kdn/\alpha \\
	& \leq (2\delta(2^k-1))^{2\delta+1}{\left(1+\frac{1}{2^{k+1}-3}\right)^{\delta^2}}\Delta^3 kdn/\alpha \\
	&\cdot \left(\left(\frac{2^{k+1}-3}{2^{k+1}-2}\right)^{2\delta-1}{\frac{1}{2^{k+1}-2}}+{\left(\frac{2^{k+1}-3}{2^{k+1}-2}\right)}+{\left(\frac{2^{k+1}-3}{2^{k+1}-2}\right)^{\delta^2}}\frac{1}{(2^{k+1}-2)^{2\delta}} \right)\\
	& \leq (2\delta(2^k-1))^{2\delta+1}{\left(1+\frac{1}{2^{k+1}-3}\right)^{\delta^2}}\Delta^3 kdn/\alpha.\tag*{\qedhere}
	\end{align*}
\end{proof}

In summary, 
$k$-\mmeans$(\otimes_{[k]},P)$ returns a clustering cost at most $(1+8\cdot2^kck^2(\Delta+1)^2)(1+\alpha)\opt_k(P)$ in $2^{O(k^2\Delta+k\Delta \log \Delta+\Delta^2-\log \alpha)}d|P|$ time with probability at least $p^kq^{k\Delta}$.
We obtain the following theorem by setting $\alpha=\epsilon/3$, and $c=\frac\alpha{8\cdot2^k k^2(\Delta+1)^2}$.
\begin{theorem}\label{thm:main}
Given a $\Delta$-missing $n$-point set $P$ in $\mathbb{H}^d$, a $(1+\epsilon)$-approximate solution to the $k$-means clustering problem can be found in $2^{O(\max\{\Delta^4k(\log\Delta+k),\ \frac {\Delta k }{\epsilon} (\log{\frac 1 \epsilon}+k)\})}dn$ time with a constant probability $1/2$.
\end{theorem}
\begin{proof}
	We prove that for a constant $\epsilon>0$ and $\Delta$-missing points $P$ of $\mathbb{H}^d$, we can find a $(1+\epsilon)$-approximation of $k$-means clustering of $P$ in $2^{O(\max\{\Delta^4k(\log\Delta+k),\ \frac {\Delta k }{\epsilon} (\log{\frac 1 \epsilon}+k)\})}d|P|$ time with a constant probability.
	We denote the time complexity as $T\cdot dn$, where $n=|P|$. As mentioned in Section~\ref{sec:k-without-oracle}, 
	$k$-\mmeans$(\otimes_{[k]},P)$ returns a clustering cost at most $(1+8\cdot2^kck^2(\Delta+1)^2)(1+\alpha)\opt_k(P)$ in $2^{O(T)}d|P|$ time with probability at least $p^kq^{k\Delta}$, where $T = k^2\Delta+k\Delta \log \Delta+\Delta^2-\log \epsilon-k\log p-k\Delta\log q$. 
	
	\medskip
	We set $\alpha=\epsilon/3$, and $c=\frac\alpha{8\cdot2^k k^2(\Delta+1)^2}$. Therefore, we have the followings for $\lambda$ which defined in Lemma~\ref{lem:initial};
	\begin{itemize}
		\item {$-\log c=O(k+\log \Delta-\log \epsilon)$}, and
		\item {$\lambda\in O(\max\{{\epsilon^{-1/2\Delta}},\Delta^{3/2\Delta}\})$}
	\end{itemize}
	
	We consider two cases: Case~(i) $\Delta^{3/2\Delta}\in O({\epsilon^{-1/2\Delta}})$ and Case~(ii) ${\epsilon^{-1/2\Delta}}\in O(\Delta^{3/2\Delta})$.
	\subparagraph{Case (i) $\Delta^{3/2\Delta}\in O({\epsilon^{-1/2\Delta}})$ :}
	In this case, $\lambda\in O(\epsilon^{-1/2\Delta})\subset O(1/\epsilon)$ for $\Delta\geq 1$. Thus, we deduce the followings for the probailities $q$ and $p$ defined in Lemmas~\ref{lem:superset} and~\ref{lem:initial};
	\begin{itemize}
		\item {$-\log p\in O(\Delta\lambda(\log\Delta-\log c))\subset O(\frac \Delta { \epsilon}(k-\log \epsilon)$}
		\item $-\log q \in O(\frac 1 \alpha (\log \Delta -\log c \alpha))\subset O(\frac 1 { \epsilon} (k-\log \epsilon))$.
	\end{itemize}
	Thus, 
	\[
	T\in O\left(\frac {\Delta k}{ \epsilon}(k-\log \epsilon)\right).
	\]
	\subparagraph{Case~(ii) ${\epsilon^{-1/2\Delta}}\in O(\Delta^{3/2\Delta})$ :}
	In such case, we have $\lambda\in O(\Delta^{3/2\Delta})\subset O(\Delta^{1.5})$;
	\begin{itemize}
		\item {$-\log p\in O(\Delta\lambda(\log\Delta-\log c))\subset O(\Delta^{2.5}(k+\log \Delta)$}
		\item $-\log q \in O(\frac 1 \alpha(\log \Delta -\log c \alpha))\subset O(\Delta^3 (k+\log \Delta))$.
	\end{itemize}
	Therefore, we have;
	\[
	T\in O(\Delta^4k(k+\log \Delta)).
	\]
	
	\medskip
	Finally, we have
	\begin{equation*}
	2^{O(T)}\cdot dn\in2^{O(\max\{\Delta^4k(\log\Delta+k),\ \frac {\Delta k }{\epsilon} (\log{\frac 1 \epsilon}+k)\})}dn. \tag*{\qedhere}
	\end{equation*}
\end{proof}
\bibliography{paper}{}

\newpage
\appendix

\section{Hardness Result}\label{apd:hardness}
Our algorithm is almost tight in the sense
that it is exponential in both $k$ and $\Delta$ but linear in both $n$ and $d$.
\begin{theorem}[\cite{lee2013clustering}]
	For fixed $k\geq 3$ and $\epsilon>0$, there is no algorithm
	that computes an $(1+\alpha)$-approximate $k$-means clustering in
	time polynomial of $n, d$ and $\Delta$ unless P=NP.
\end{theorem}
\begin{proof}
	We reduce \textsc{graph $k$-coloring} to the $k$-means clustering problem.
	Given $G=(V,E)$ with $V=\{v_1,\ldots, v_n\}$ and $E=\{(v_{1,1},v_{1,2}) ,\ldots, (v_{m,1},v_{m,2})  \}$, 
	we construct a set $S$ of $n$ points of $\mathbb{H}^m$ as follows.
	Each point $p_t$ in $S$ corresponds to each vertex $v_t$ for each $t\in[n]$,
	and the $i$th coordinate corresponds to $(v_{i,1}, v_{i,2})$ for each $i\in[m]$. 
	For each $t\in[n]$ and $i\in[m]$, let 
	\[ (p_t)_i=\left\{\begin{array}{ll} -1 & \textnormal{if }  v_t=v_{i,1},
	\\ +1 &\textnormal{if } v_t=v_{i,2}, \\
	\otimes &\textnormal{ otherwise. }
	\end{array}\right.\]
	
	Assume that $G$ is $k$-colorable if and only if
	$S$ can be partitioned into $k$ subsets such that no two points in the same subset 
	have fixed values -1 and +1 in the same coordinates.
	The $k$-means clustering cost of such a partition is $0$ by definition. 
	The number of points in $S$ is $n$, the number of coordinates is $m$,
	and the number of missing entries for each point is at most $m$.  	
	If there is an 
	algorithm that computes an $(1+\epsilon)$-approximate $k$-means clustering in time polynomial in 
	$n, d$ and $\Delta$, we 
	can solve \textsc{graph $k$-coloring} in time polynomial in $n$ and $m$, implying P = NP.
\end{proof}

\begin{theorem}[\cite{awasthi2015hardness}]
	For fixed $\Delta\geq 0$ and $\epsilon>0$, there is no algorithm 
	that computes an $(1+\epsilon)$-approximate $k$-means clustering in
	time polynomial of $n, d$ and $k$ unless P=NP.
\end{theorem}


\end{document}